\documentclass[10pt,journal]{IEEEtran}

\usepackage{booktabs}
\usepackage{array}
\ifCLASSOPTIONcompsoc
  \usepackage[nocompress]{cite}
\else
  \usepackage{cite}
\fi
\usepackage{amsmath,amssymb,amsfonts}
\usepackage{mathrsfs}
\usepackage{graphicx}
\ifCLASSOPTIONcompsoc
  \usepackage[caption=false,font=footnotesize,labelfont=sf,textfont=sf]{subfig}
\else
  \usepackage[caption=false,font=footnotesize]{subfig}
\fi
\usepackage{textcomp}
\usepackage{xcolor}
\usepackage{url}
\usepackage{orcidlink}
\hypersetup{hidelinks}
\usepackage[linesnumbered,ruled,vlined]{algorithm2e}

\usepackage{enumitem}

\usepackage{amsthm}

\usepackage{cleveref}

\newtheorem{lemma}{Lemma}

\newcommand{\inlineheading}[1]{\vspace{2mm}\noindent\textbf{#1}\par\vspace{1mm}}
\begin{document}

\title{Maximizing Parallel Execution of Series-Parallel Task Graphs for Safety-Critical Embedded Control}

\author{Jinghao~Sun~\orcidlink{0000-0003-4107-1123},
        Zhenchu~Hu~\orcidlink{0009-0000-7008-6055},
        Ye~Ma~\orcidlink{0000-0002-2086-5648},
        Bo~Tang,
        Qingxu~Deng~\orcidlink{0000-0002-5185-6306},
        and~Xiuzhen~Cheng~\orcidlink{0000-0001-9998-2398}%
\thanks{Jinghao Sun and Xiuzhen Cheng are with the School of Computer Science
and Technology, Shandong University, Qingdao 266237, China
(e-mail: jhsun@sdu.edu.cn; xzcheng@sdu.edu.cn).}%
\thanks{Zhenchu Hu and Ye Ma are with the School of Computer
Science and Technology, Dalian University of Technology, Dalian 116024, China
(e-mail: huzhenchu@mail.dlut.edu.cn; yema@dlut.edu.cn).}%
\thanks{Bo Tang is with Weichai Power Co., Ltd., Weifang 261061, China. (e-mail: b.tang@hotmail.com)}%
\thanks{Qingxu Deng is with the School of Computer Science and Engineering,
Northeastern University, Shenyang 110004, China
(e-mail: dengqx@mail.neu.edu.cn).}%
\thanks{Corresponding author: Ye Ma (e-mail: yema@dlut.edu.cn).}}

\maketitle

\begin{abstract}
Safety-critical embedded control programs must complete each control cycle
within a bounded period. Sequential execution on conventional processors can
become a bottleneck when the dependency structure of the program contains
subtasks that could be executed concurrently. This paper studies the Maximum
Parallel Execution (MPE) problem for series-parallel task graphs under a staged
batching model: compatible tasks inside one batch execute in parallel, while
the selected batches are launched sequentially in a topological order that
preserves precedence. We formulate MPE as a weighted clique-partitioning problem
that minimizes the sum of batch execution times, with each batch cost determined
by its slowest task. To solve this problem efficiently, we propose a
Lagrangian-based Iterative Heuristic (LIH). LIH constructs a pricing-filtered
restricted pool of feasible candidate batches from singleton columns and random
greedy clique generation. It then applies Lagrangian pricing to guide column
selection and uses a repair procedure to recover a legal clique partition.
Experiments against
a weighted mixed-graph-coloring branch-and-bound baseline and a randomized
greedy baseline show that LIH matches the exact optimum in \(91.25\%\) of
comparable instances, with an average gap of \(0.073\%\) and an average runtime
of \(18.19\) ms. In the largest exact-reference node setting, the exact
baseline requires hundreds of seconds on average, whereas LIH remains below
\(50\) ms. We further present an end-to-end PLC ladder-logic case study in
which PLCOpen-style programs are converted to MPE graphs, optimized by LIH,
translated into FPGA-oriented HDL, and simulated against the original PLC scan
execution.
\end{abstract}


\IEEEpeerreviewmaketitle

\section{Introduction}
\label{sec:intro}

A series-parallel (SP) task graph is a structured special case of a directed
acyclic graph (DAG). It models a computation as subtasks composed by two
fundamental operations: series composition, which chains subtasks that must
execute sequentially due to data dependencies, and parallel composition, which
groups independent subtasks that may execute concurrently. This recursive
structure arises naturally in a broad class of safety-critical embedded control
programs---PID regulation loops, sensor fusion pipelines, and fault-detection
logic---where certain operations depend on the output of others while the
remaining operations are mutually independent
\cite{valdes1979recognition,keller2011peelsched}. Given such a task graph, a
central question is: how should its subtasks be dispatched to maximize parallel
execution and thereby minimize the overall completion time (makespan)? We
refer to this as the Maximum Parallel Execution (MPE) problem on SP task
graphs, and it is the subject of this paper.

The MPE problem is motivated by timing constraints in real-time embedded
control. Safety-critical control systems---found in industrial automation,
robotics, autonomous vehicles, and high-precision instrumentation, must
complete their control algorithms within a bounded control period to guarantee
stability and responsiveness \cite{nilsson1998real,henriksson2005optimal}.
Traditionally, these control tasks are executed sequentially on
microcontrollers or digital signal processors, and decades of effort have been
devoted to optimizing serial performance through instruction-level pipelining,
compiler optimizations, and clock-frequency scaling
\cite{dhanabalan2022scan,hennig2016towards,oliveira2024shared}. These serial
techniques are increasingly constrained by power, frequency, and memory latency
limits. At the same time, modern applications often demand shorter control
periods, for example when faster physical dynamics or finer-grained actuation
require a loop to run at millisecond or sub-millisecond scale. Exposing
task-level parallelism is therefore an important way to reduce the time spent
in one control cycle.

Field-Programmable Gate Arrays (FPGAs) are a natural platform for this form of
parallelization. Prior implementations of industrial control, predictive
current control, model predictive control, and IEC~61131-3 timer logic show
that typical embedded control functions can be mapped to FPGA hardware with
moderate resource usage and deterministic timing behavior
\cite{monmasson2007fpga,sankar2021fpga,lucia2017optimized,chmiel2023fpga}.
Beyond their familiar role in neural-network acceleration
\cite{mittal2020survey,nurvitadhi2017can}, FPGAs are attractive for embedded
control because they combine hardware-level determinism, reconfigurability,
low power consumption, and robustness in harsh operating environments
\cite{monmasson2007fpga,bernardeschi2015sram}. Under this resource-abundant
setting, the main question is no longer processor allocation, but structure:
which tasks can be batched together without violating dependencies, and in what
staged order should those batches be executed?

Compared with scheduling on arbitrary DAGs, a well-studied but NP-hard problem
in general \cite{ref_dag_nphard}, the SP restriction admits significantly more
efficient solutions. SP graphs possess a recursive decomposition tree that
mirrors their series-parallel construction, and this structure can be exploited
algorithmically \cite{valdes1979recognition,keller2011peelsched,ref_messi}.
Yet, existing work on SP graph scheduling has largely focused on
resource-constrained settings, such as limited processors, or general-purpose
parallel platforms, and does not directly address the FPGA-oriented scenario
where computation units are abundant but the goal is to maximize the
parallelism exposed from an originally sequential program
\cite{keller2011peelsched,ref_messi}. Moreover, prior studies on FPGA-based
control acceleration typically treat the task graph as a flat pipeline or rely
on ad hoc manual partitioning, without systematically exploiting the SP
structure inherent in control programs \cite{dhanabalan2022scan}.

This paper formalizes the MPE problem on SP task graphs for FPGA platforms
under the assumption of sufficient computation units, and proposes an efficient
\textsc{Lagrangian-based Iterative Heuristic} (LIH) for generating static
staged batching plans. Given a control program with an SP-like dependency
structure, LIH exposes within-stage parallelism while preserving precedence
consistency between stages.

The main contributions are as follows:
\begin{itemize}
    \item We formulate MPE on an augmented series-parallel task graph that
    combines precedence arcs, compatibility edges, weighted task costs, and
    staged batching semantics for FPGA-oriented execution.
    \item We express MPE as a weighted clique-partitioning problem with
    contracted-graph acyclicity, giving both an exact column model and a direct
    interpretation in which each selected clique is one parallel hardware
    stage.
    \item We propose LIH, which uses pricing-filtered candidate generation,
    Lagrangian reduced-cost pricing, feasibility repair, and local refinement
    to avoid enumerating the exponential space of feasible batches.
    \item We evaluate LIH against an exact weighted mixed-graph-coloring
    baseline and a randomized greedy baseline, and further demonstrate the
    complete PLC ladder-to-HDL workflow through parsing, MPE optimization,
    Verilog generation, simulation, and scan-cycle comparison.
\end{itemize}

The remainder of this paper is organized as follows. Section~\ref{sec:model}
presents the system model and problem formulation.
Section~\ref{sec:method7-lagrangian} describes the proposed
Lagrangian-based iterative heuristic. Section~\ref{sec:exp} reports
experimental results. Section~\ref{sec:case-study} shows how a PLC ladder
program is mapped to the MPE graph and how a batching plan can be realized on
FPGA-oriented hardware. Related work is discussed in Section~\ref{sec:related}.
Section~\ref{sec:conclusion} concludes the paper.

\section{System Model and Problem Formulation}
\label{sec:model}

This section formulates Maximum Parallel Execution (MPE) as a
clique-partitioning problem on an augmented series-parallel task graph (SPG).
We first define the graph model and staged batching semantics, and then give
the exact integer formulation used as the reference problem.

\subsection{Series-Parallel Task Graph}

An SPG describes a computation built from series and parallel composition.
Series composition imposes an execution order, whereas parallel composition
places subgraphs in the same logical region without an internal ordering
relation. This structure appears naturally in embedded control programs, where
serial data dependencies coexist with independent sensing, filtering,
diagnosis, and actuation branches that can be mapped to parallel FPGA logic.

Let the input SPG be represented as a mixed graph\footnote{Strictly speaking,
the model used in this paper is an augmented SPG rather than a conventional SP
graph. The series-parallel restriction is imposed on the directed precedence
component \((V,A)\); the model additionally includes undirected compatibility
edges \(E\), vertex execution times \(w_v\), and the staged-batching semantics
with contracted-graph acyclicity used for MPE.}
\[
    G=(V,E,A),
\]
where \(V\) is the set of computational tasks, \(E\) is the set of undirected
compatibility edges, and \(A\) is the set of directed precedence arcs. Each
vertex \(v\in V\) represents an atomic operation, a function block, or a
higher-level control action, and has a positive execution time \(w_v>0\). The
directed graph \((V,A)\) is acyclic and follows the series-parallel precedence
structure. A directed arc \((u,v)\in A\) means that task \(u\) must finish
before task \(v\) can start. An undirected edge \(\{u,v\}\in E\) means that
\(u\) and \(v\) are mutually compatible and may be placed in the same parallel
execution batch.

The compatibility relation collects implementation conditions that cannot be
represented by precedence alone, such as conflicting state updates, exclusive
I/O access, or hardware resources that cannot be duplicated. Resource
attributes such as LUT, DSP, BRAM, memory-port, or safety-partition usage can
also be encoded in \(E\). In this paper we assume sufficient FPGA computation
units for any selected batch, so the limiting factor is the dependency and
compatibility structure of \(G\). Under this concurrent execution semantics,
the duration of a selected batch is determined by its slowest task.

\subsection{Clique Partitioning and Parallelization Semantics}

Partitioning turns graph-level parallelism into an executable sequence of
parallel batches. Write \(u\leadsto_A v\) if there is a directed path from
\(u\) to \(v\) in \((V,A)\). A subset \(S\subseteq V\) is a feasible batch if
its vertices are pairwise compatible and have no direct or transitive
precedence relation:
\[
    \{u,v\}\in E,\quad
    u\not\leadsto_A v,\quad
    v\not\leadsto_A u,
    \qquad
    \forall u,v\in S,\ u\neq v .
\]
Equivalently, \(S\) is a clique in the compatibility graph that does not
collapse any precedence path. The execution cost of \(S\) is
\begin{equation*}
    c(S)=\max_{v\in S} w_v .
\end{equation*}
In the FPGA deployment, one batch corresponds to one execution stage: all tasks
in \(S\) are enabled together, and inter-stage registers hold values required
by later stages. Compared with serial execution, the local saving of a
non-singleton batch is \(\sum_{v\in S}w_v-\max_{v\in S}w_v\). Singleton
batches remain available for tasks that cannot be legally or profitably merged.

A clique partition \(\mathcal{C}=\{C_1,\ldots,C_m\}\) assigns every vertex to
exactly one feasible batch:
\[
    C_i\cap C_j=\emptyset\ (i\neq j),
    \qquad
    \bigcup_{i=1}^{m} C_i=V .
\]
We use staged sequential batching semantics. Tasks in the same batch execute
concurrently, while batches are issued one at a time by a controller according
to a topological order of the contracted graph. This preserves the sequential
effects of the original program while exploiting spatial parallelism inside
each stage.

After each \(C_i\) is contracted into a super-vertex, every original arc
\((u,v)\in A\) with \(u\in C_i\), \(v\in C_j\), and \(i\neq j\) induces an arc
\(C_i\rightarrow C_j\). This contraction replaces several fine-grained tasks
by one coarse-grained parallel stage. The partition is legal only if the
contracted graph \(G/\mathcal{C}\) is acyclic, which ensures that the selected
batches admit at least one staged execution order. Fig.~\ref{fig:shrink}
shows an infeasible partition whose contraction creates a directed cycle.

\begin{figure}[htb]
    \centering
    \includegraphics[width=\linewidth]{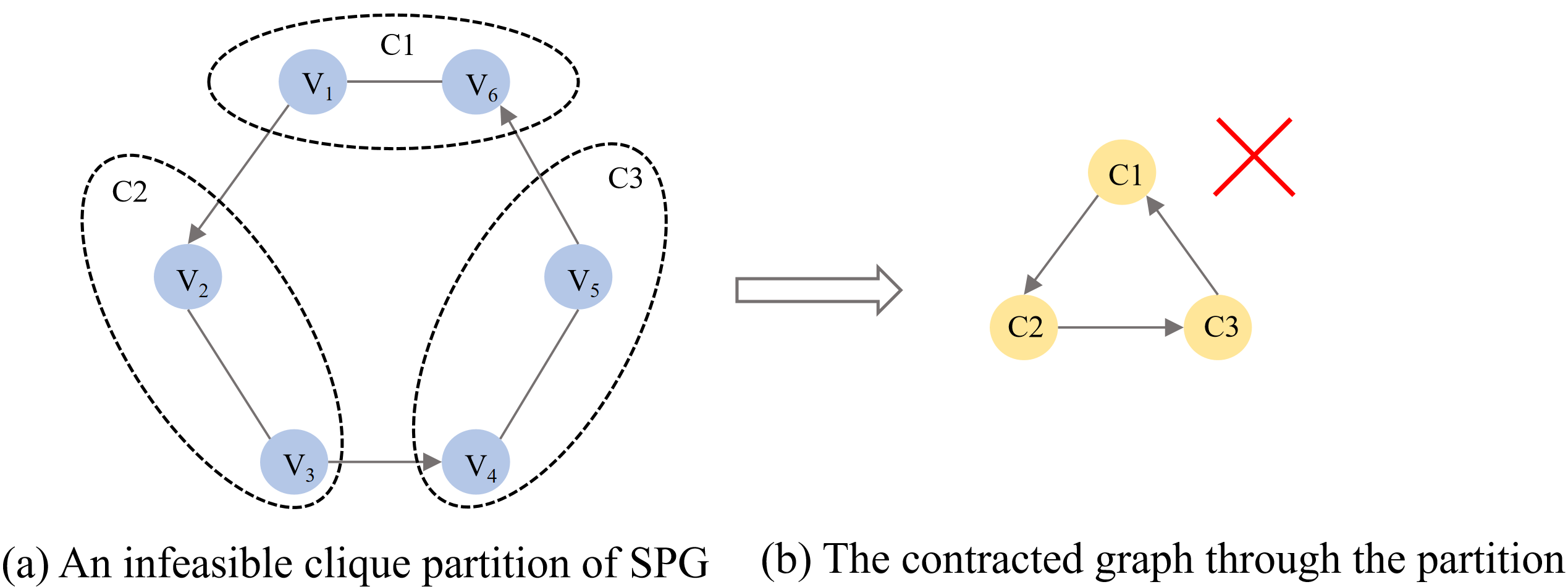}
    \caption{An infeasible clique partition whose contracted graph contains a directed cycle.}
    \label{fig:shrink}
\end{figure}

For a legal partition \(\mathcal{C}\), the makespan under the staged batching
semantics is the sum of the selected batch costs:
\begin{equation}
    T(\mathcal{C})
    =
    \sum_{C\in\mathcal{C}} c(C)
    =
    \sum_{C\in\mathcal{C}}\max_{v\in C}w_v .
    \label{eq:model-partition-cost}
\end{equation}
Thus, maximizing parallel execution is equivalent to minimizing the additive
cost of legal batches. In Fig.~\ref{fig:spg-partition}, for example, the
partition \(C_1=\{v_1,v_5\}\) and \(C_2=\{v_2,v_3,v_4\}\) has costs
\(15\) and \(25\), giving \(T(\mathcal{C})=40\).

\begin{figure}[!t]
    \centering
    \includegraphics[width=\linewidth]{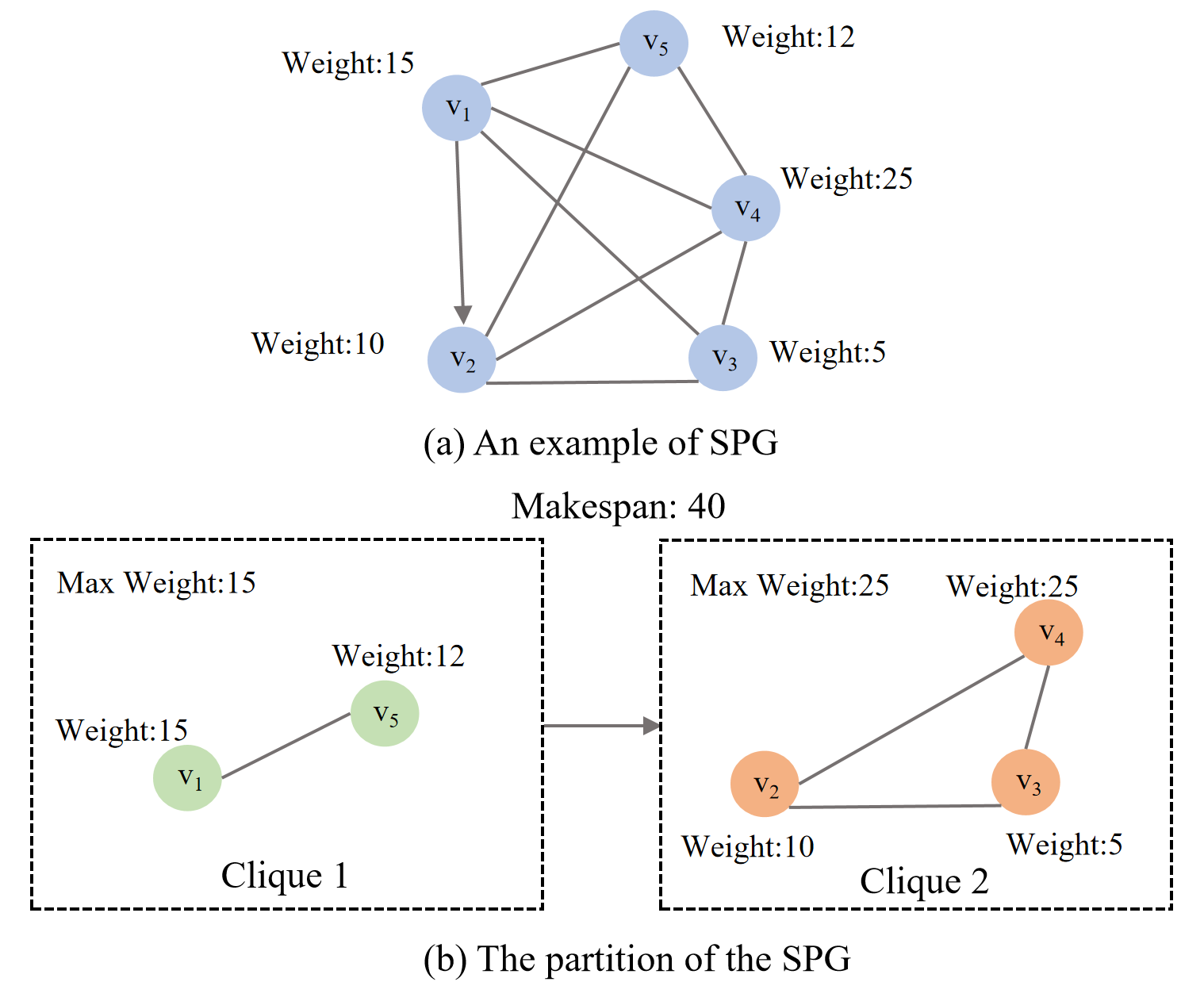}
    \caption{An example SPG and a feasible clique partition.}
    \label{fig:spg-partition}
\end{figure}

\subsection{Maximum Parallel Execution Formulation}

Let \(\mathcal{P}^{\star}\) be the unrestricted set of all feasible batches in
\(G\). It contains every feasible clique satisfying compatibility and
no-internal-path conditions. Fig.~\ref{fig:pstar-example} illustrates
\(\mathcal{P}^{\star}\) for the same SPG.

\begin{figure}[!t]
    \centering
    \includegraphics[width=0.86\linewidth]{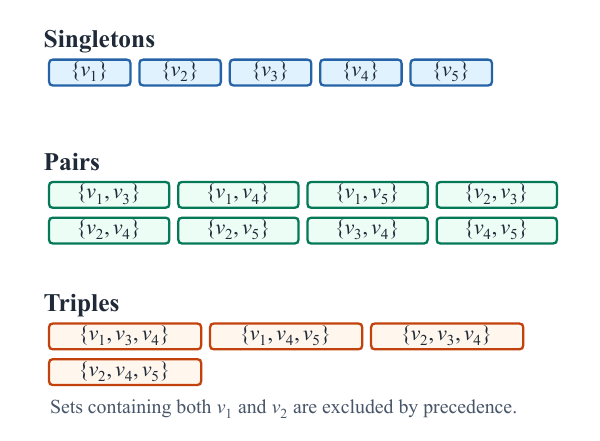}
    \caption{Example of the unrestricted feasible batch set
    \(\mathcal{P}^{\star}\) for the SPG in Fig.~\ref{fig:spg-partition}.}
    \label{fig:pstar-example}
\end{figure}

Each candidate batch is represented as a \emph{column}
\(p\in\mathcal{P}^{\star}\) with vertex set \(S_p\subseteq V\), cost
\begin{equation*}
    c_p=\max_{v\in S_p}w_v ,
\end{equation*}
and incidence coefficient
\begin{equation*}
    a_{vp} =
    \begin{cases}
        1, & \text{if } v\in S_p,\\
        0, & \text{otherwise}.
    \end{cases}
\end{equation*}
The binary variable \(z_p\) indicates whether column \(p\) is selected:
\begin{equation*}
    z_p =
    \begin{cases}
        1, & \text{if } S_p \text{ is selected as a batch},\\
        0, & \text{otherwise}.
    \end{cases}
\end{equation*}

The Maximum Parallel Execution (MPE) problem is then formulated as
\begin{align}
    T^{\star}(G)=
    \min_{z}\quad
        & \sum_{p\in\mathcal{P}^{\star}} c_p z_p
        \label{eq:model-mpe-obj} \\
    \text{s.t.}\quad
        & \sum_{p\in\mathcal{P}^{\star}} a_{vp}z_p = 1,
          && \forall v\in V,
        \label{eq:model-mpe-cover} \\
        & G/\{S_p:z_p=1\}\text{ is acyclic},
        \label{eq:model-mpe-acyc} \\
        & z_p\in\{0,1\},
          && \forall p\in\mathcal{P}^{\star}.
        \notag
\end{align}
Constraint~\eqref{eq:model-mpe-cover} is the exact coverage constraint: every
task must be assigned to one and only one selected batch.
Constraint~\eqref{eq:model-mpe-acyc} preserves global precedence consistency
after batch contraction. The objective~\eqref{eq:model-mpe-obj} minimizes the
total execution time of the selected parallel batches.

\subsection{Challenges for Parallelization}

Although exact, the formulation is difficult to solve directly. It contains
set-partitioning and clique-partitioning structure, which is NP-hard in
general~\cite{ref_dag_nphard}, and the column set
\(\mathcal{P}^{\star}\) can be exponentially large. The selected columns must
also cover every vertex exactly once and yield an acyclic contracted graph, so
local clique feasibility alone is insufficient.

The next section therefore introduces a Lagrangian-based iterative heuristic.
Instead of enumerating all of \(\mathcal{P}^{\star}\), it works with a
restricted candidate set \(\mathcal{P}\subseteq\mathcal{P}^{\star}\) while
retaining the same column notation \(S_p\), \(c_p\), \(a_{vp}\), and \(z_p\).
The exact coverage equality~\eqref{eq:model-mpe-cover} then becomes the source
of the Lagrangian multipliers.

\section{Lagrangian-based Iterative Heuristic}
\label{sec:method7-lagrangian}

The exact model in Section~\ref{sec:model} uses the unrestricted candidate set
\(\mathcal{P}^{\star}\), which can be exponentially large. The
Lagrangian-based Iterative Heuristic (LIH) instead works over a restricted
pool \(\mathcal{P}\subseteq\mathcal{P}^{\star}\). Each column in
\(\mathcal{P}\) is still a locally feasible clique; the restriction only
limits which feasible cliques are considered.

LIH first builds a preliminary pool from singleton and random-greedy cliques.
Because the useful cliques are unknown a priori, a warm-up Lagrangian pricing
pass ranks the generated non-singleton columns. LIH keeps all singleton columns
and only the top-ranked non-singleton columns for the main iterations, reducing
the influence of low-value cliques while preserving a trivial feasible
partition. The retained pool is then priced, repaired, and refined; repair is
needed because the relaxed selection may overlap, miss vertices, or create a
cycle after contraction.

For clarity, Table~\ref{tab:m7-notation} lists the notation used in this
section.

\begin{table*}[!t]
\caption{LIH notation.}
\label{tab:m7-notation}
\centering
\setlength{\tabcolsep}{4pt}
\renewcommand{\arraystretch}{1.06}
\begin{tabular*}{\textwidth}{@{\extracolsep{\fill}}>{\raggedright\arraybackslash}p{0.17\textwidth}
                >{\raggedright\arraybackslash}p{0.29\textwidth}
                >{\raggedright\arraybackslash}p{0.17\textwidth}
                >{\raggedright\arraybackslash}p{0.29\textwidth}@{}}
\toprule
Symbol & Meaning & Symbol & Meaning \\
\midrule
\(G=(V,E,A)\) & SPG: tasks, compatibility edges, and precedence arcs. & \(\lambda_v\), \(\lambda\), \(\lambda_v^t\), \(\lambda^{\mathrm{w}}\) & Vertex multiplier, vector, iteration-\(t\) value, and warm-up vector. \\
\(u,v\) & Task vertices. & \(L(z,\lambda)\) & Lagrangian objective. \\
\(S\), \(S_p\) & Vertex subset; subset of column \(p\). & \(\bar{c}_p(\lambda)\), \(\bar{c}(\lambda)\) & Column reduced cost and vector. \\
\(w_v\), \(w\), \(w_{\max}\) & Vertex time, weight vector, maximum weight. & \(\phi(\lambda)\), \(\phi_t\) & Dual function and current dual value. \\
\(\mathcal{P}^{\star}\), \(\mathcal{P}\) & All feasible columns; restricted pool. & \(\mathcal{P}^{-}(\lambda)\), \(\mathcal{P}^{-}\) & Negative-reduced-cost columns. \\
\(\mathcal{P}_0\), \(\mathcal{P}_{\mathrm{sg}}\), \(\mathcal{P}_{\mathrm{rg}}\) & Preliminary, singleton, and random-greedy pools. & \(h_v(\lambda)\), \(h_v\) & Relaxed coverage count. \\
\(\eta\) & Fraction of non-singleton columns retained after warm-up pricing. & \(g_v(\lambda)\), \(g_v\), \(g^t\) & Coverage violation/subgradient. \\
\(p\), \(m_0\), \(m\), \(n\), \(q\) & Column index, preliminary columns, retained columns, vertices, and maximum column size. & \(\alpha_t\), \(\rho\), \(\epsilon\) & Step size, scale, and minimum step. \\
\(m_{\max}\), \(R_{\mathrm{rg}}\), \(\xi\), \(\Delta\) & Pool cap, random-greedy rounds, seed, and optional deadline. & \([x]_{\epsilon}^{w_{\max}}\), \(x\) & Scalar clipping to \([\epsilon,w_{\max}]\). \\
\(\ell\), \(V^{\downarrow}\), \(V_{\pi}\) & Random-greedy round, ordered list, and shuffled list. & \(M\), \(\Pi_{[-M,M]}\) & Multiplier bound and projection. \\
\(C\), \(\mathcal{C}\) & Clique/batch and clique partition. & \(UB_t\), \(UB\) & Best feasible objective value. \\
\(\mathcal{C}^{0}\), \(\widehat{\mathcal{C}}\), \(\mathcal{C}^{best}\) & Initial, repaired, and best partitions. & \(\mathrm{saving}(p)\), \(\tau(p)\), \(\mathrm{score}(p)\) & Repair saving, insertion order, and ranking key. \\
\(c(S)\), \(c(C)\), \(c_p\) & Batch or column cost. & \(t\), \(T\), \(h\) & Iteration index, limit, and repair interval. \\
\(a_{vp}\) & 1 if \(v\in S_p\), 0 otherwise. & \(\|\cdot\|_2\), \(|\cdot|\), \(O(\cdot)\) & Norm, cardinality, complexity. \\
\(z_p\), \(z_p(\lambda)\) & Column selection and relaxed selection. & \(\emptyset\), \(\mathbb{R}\), \(\{0,1\}\) & Empty, real, and binary sets. \\
\(G/\{S_p:z_p=1\}\), \(G/\mathcal{C}\) & Contracted graph. & & \\
\bottomrule
\end{tabular*}
\end{table*}

Lagrangian relaxation moves coupling constraints into the objective with
multipliers~\cite{geoffrion1974lagrangean,fisher1981lagrangian}. In LIH, the
multipliers act as vertex prices: over-covered vertices become more expensive,
whereas uncovered vertices become cheaper. The pricing step therefore provides
both a dual search signal and a ranking rule for candidate cliques.

\subsection{Restricted Column Model and Relaxation}
\label{subsec:m7-column-model}

We reuse the notation introduced in Section~\ref{sec:model}. The only change
from the exact formulation is the candidate set:
\[
    \mathcal{P}\subseteq\mathcal{P}^{\star},
    \qquad
    \{\{v\}:v\in V\}\subseteq\mathcal{P}.
\]
For each \(p\in\mathcal{P}\), the vertex set \(S_p\), cost \(c_p\), incidence
coefficient \(a_{vp}\), and binary selection variable \(z_p\) have the same
meaning as in Section~\ref{sec:model}. The preliminary pool can be written as
\begin{equation*}
    \mathcal{P}_0
    =
    \mathcal{P}_{\mathrm{sg}}
    \cup
    \mathcal{P}_{\mathrm{rg}},
\end{equation*}
where \(\mathcal{P}_{\mathrm{sg}}\) contains singleton columns and
\(\mathcal{P}_{\mathrm{rg}}\) contains random greedy cliques. LIH then applies
a warm-up pricing screen to obtain the working pool
\begin{equation*}
    \mathcal{P}
    =
    \mathcal{P}_{\mathrm{sg}}
    \cup
    \operatorname{Top}_{\eta}
    \bigl(\mathcal{P}_0\setminus\mathcal{P}_{\mathrm{sg}}\bigr),
\end{equation*}
where \(\operatorname{Top}_{\eta}(\cdot)\) denotes the top \(\eta\) fraction
of non-singleton columns under the warm-up reduced-cost ranking. The singleton
columns ensure that the restricted model always has at least the serial
feasible partition.

Using these columns, the restricted column-selection model is formulated as
\begin{align}
    \min_{z}\quad
        & \sum_{p\in\mathcal{P}} c_p z_p
        \notag \\
    \text{s.t.}\quad
        & \sum_{p\in\mathcal{P}} a_{vp}z_p = 1,
          && \forall v\in V,
        \label{eq:m7-rmp-cover} \\
        & G/\{S_p:z_p=1\}\text{ is acyclic},
        \label{eq:m7-rmp-acyc} \\
        & z_p\in\{0,1\},
          && \forall p\in\mathcal{P}.
        \notag
\end{align}
Constraint~\eqref{eq:m7-rmp-cover} enforces exact coverage, and
constraint~\eqref{eq:m7-rmp-acyc} enforces acyclicity after contraction. Since
\(\mathcal{P}\) is only a subset of \(\mathcal{P}^{\star}\), the restricted
optimum may be worse than the unrestricted optimum; LIH accepts this trade-off
to avoid the full exponential column space.

\inlineheading{Relaxation Strategy}

The local feasibility of each column is guaranteed during candidate generation:
LIH checks pairwise compatibility and rejects cliques containing a directed
precedence path. Coverage, however, is global. The equality
\begin{equation*}
    \sum_{p\in\mathcal{P}} a_{vp}z_p=1
\end{equation*}
couples all columns that contain vertex \(v\).

LIH therefore relaxes the exact coverage constraints and assigns one multiplier
\(\lambda_v\) to each vertex. If the relaxed solution covers \(v\) more than
once, \(\lambda_v\) is increased, making columns containing \(v\) more
expensive. If the relaxed solution does not cover \(v\), \(\lambda_v\) is
decreased, making columns containing \(v\) more attractive. Thus, the
multipliers are adaptive vertex prices derived directly from
\eqref{eq:model-mpe-cover}.

The pricing step is therefore a search guide rather than a complete feasibility
mechanism: coverage and contracted acyclicity are restored later by repair,
while the relaxed subproblem remains separable by columns.

\subsection{Lagrangian Pricing and Multiplier Update}
\label{subsec:m7-lagrangian-derivation}

Let \(\lambda_v\in\mathbb{R}\) be the Lagrange multiplier associated with the
coverage equality of vertex \(v\). For a fixed multiplier vector
\begin{equation*}
    \lambda=(\lambda_v)_{v\in V},
\end{equation*}
the Lagrangian objective is obtained by adding the weighted coverage violation
to the original objective:
\begin{align}
    L(z,\lambda)
    &=
    \sum_{p\in\mathcal{P}} c_p z_p
    +
    \sum_{v\in V}\lambda_v
    \left(
    \sum_{p\in\mathcal{P}}a_{vp}z_p - 1
    \right)
    \notag \\
    &=
    \sum_{p\in\mathcal{P}} c_p z_p
    +
    \sum_{v\in V}
    \sum_{p\in\mathcal{P}}
        \lambda_v a_{vp}z_p
    -
    \sum_{v\in V}\lambda_v .
    \label{eq:m7-lagrangian-split}
\end{align}

Exchanging the summation order gives
\begin{align*}
    \sum_{v\in V}
    \sum_{p\in\mathcal{P}}
        \lambda_v a_{vp}z_p
    &=
    \sum_{p\in\mathcal{P}}
    \sum_{v\in V}
        \lambda_v a_{vp}z_p .
\end{align*}
Since \(a_{vp}=1\) only when \(v\in S_p\),
\begin{align}
    \sum_{v\in V}
        \lambda_v a_{vp}z_p
    &=
    \left(
        \sum_{v\in S_p}\lambda_v
    \right)z_p .
    \label{eq:m7-incidence-rewrite}
\end{align}
Substituting~\eqref{eq:m7-incidence-rewrite} into
\eqref{eq:m7-lagrangian-split} gives
\begin{align}
    L(z,\lambda)
    &=
    -\sum_{v\in V}\lambda_v
    +
    \sum_{p\in\mathcal{P}}
    \left(
        c_p+\sum_{v\in S_p}\lambda_v
    \right)z_p .
    \label{eq:m7-lagrangian-collected}
\end{align}

Equation~\eqref{eq:m7-lagrangian-collected} defines the reduced cost of column
\(p\):
\begin{equation}
    \bar{c}_p(\lambda)
    =
    c_p+\sum_{v\in S_p}\lambda_v .
    \label{eq:m7-reduced-cost}
\end{equation}

For a fixed multiplier vector \(\lambda\), the relaxed pricing problem is
\begin{equation*}
    \min_{z\in\{0,1\}^{|\mathcal{P}|}} L(z,\lambda).
\end{equation*}
Because coverage and contracted acyclicity are relaxed during pricing, each
\(z_p\) can be optimized independently:
\begin{equation}
    z_p(\lambda)=
    \begin{cases}
        1, & \text{if } \bar{c}_p(\lambda)<0,\\
        0, & \text{if } \bar{c}_p(\lambda)\geq 0 .
    \end{cases}
    \label{eq:m7-relaxed-selection}
\end{equation}
Thus, only negative-reduced-cost columns are selected in the relaxed solution.

Substituting the best independent decisions from
\eqref{eq:m7-relaxed-selection} into \eqref{eq:m7-lagrangian-collected} yields
the Lagrangian dual function
\begin{equation}
    \phi(\lambda)
    =
    -\sum_{v\in V}\lambda_v
    +
    \sum_{p\in\mathcal{P}}
    \min\{0,\bar{c}_p(\lambda)\}.
    \label{eq:m7-dual-function}
\end{equation}
LIH does not solve the dual problem to optimality. It uses \(\phi(\lambda)\)
and the reduced costs as search signals: a small reduced cost indicates either
a low clique cost, under-covered vertices, or both.

\inlineheading{Subgradient Multiplier Update}

The dual function~\eqref{eq:m7-dual-function} is piecewise linear and may be
non-differentiable. Therefore, LIH updates the multipliers by a subgradient
method.

Let
\begin{equation*}
    \mathcal{P}^{-}(\lambda)
    =
    \{p\in\mathcal{P}:\bar{c}_p(\lambda)<0\}
\end{equation*}
be the set of columns selected by the relaxed pricing rule. The coverage count
of vertex \(v\) in this relaxed solution is
\begin{equation*}
    h_v(\lambda)
    =
    \sum_{p\in\mathcal{P}^{-}(\lambda)} a_{vp}.
\end{equation*}
The corresponding coverage violation is
\begin{equation*}
    g_v(\lambda)
    =
    h_v(\lambda)-1 .
\end{equation*}
This expression is inherited directly from the relaxed coverage equality
\(\sum_{p\in\mathcal{P}}a_{vp}z_p-1=0\), so \(g_v(\lambda)\) is a valid
subgradient component for multiplier \(\lambda_v\).

The multiplier update is
\begin{equation*}
    \lambda_v^{t+1}
    =
    \Pi_{[-M,M]}
    \left(
        \lambda_v^t+\alpha_t g_v^t
    \right),
    \quad \forall v\in V,
\end{equation*}
where \(\Pi_{[-M,M]}\) denotes projection onto the interval \([-M,M]\). In the
implementation,
\begin{equation*}
    M=2w_{\max},
    \qquad
    w_{\max}=\max_{v\in V}w_v .
\end{equation*}
The projection prevents excessively large multipliers from dominating the
original clique costs and stabilizes the search.

The update follows the price interpretation: over-covered vertices
\((h_v>1)\) become more expensive, uncovered vertices \((h_v=0)\) become
cheaper, and exactly covered vertices keep the same price.

Let \(UB_t\) be the best feasible objective value found so far, and let
\begin{equation*}
    \phi_t=\phi(\lambda^t).
\end{equation*}
When \(\phi_t<UB_t\), LIH uses a Polyak-style step size:
\begin{equation}
    \alpha_t
    =
    \left[
        \rho\frac{UB_t-\phi_t}{\|g^t\|_2^2}
    \right]_{\epsilon}^{w_{\max}},
    \quad \text{if } \phi_t<UB_t,
    \label{eq:m7-polyak-step}
\end{equation}
where \(\rho>0\) is a step-scale parameter and
\begin{equation*}
    [x]_{\epsilon}^{w_{\max}}
    =
    \min\{w_{\max},\max\{\epsilon,x\}\}.
\end{equation*}
Here, \(\epsilon>0\) is a small minimum step-size parameter. It provides the
lower clipping bound, so an active multiplier update is not reduced to an
exactly zero step by the step-size rule.

The numerator \(UB_t-\phi_t\) measures the current primal-dual gap, while
\(\|g^t\|_2^2\) normalizes the update by the coverage violation.

When \(\phi_t\geq UB_t\), LIH switches to a damped fallback step:
\begin{equation}
    \alpha_t
    =
    \left[
        \rho
        \frac{w_{\max}}
        {\sqrt{t}\max\{1,\|g^t\|_2^2\}}
    \right]_{\epsilon}^{w_{\max}} .
    \label{eq:m7-damped-step}
\end{equation}
This fallback keeps the multiplier update active, while the factor
\(1/\sqrt{t}\) gradually reduces oscillation in later iterations.

\subsection{Candidate Pool, Repair, and Local Refinement}
\label{subsec:m7-column-generation}

The candidate pool is constructed once before the main Lagrangian iterations.
Singleton columns guarantee coverage and define the serial initial partition
\(\mathcal{C}^0=\{\{v\}:v\in V\}\). Algorithm~\ref{alg:m7-random-greedy} adds
locally feasible multi-vertex cliques by checking compatibility, precedence
paths, size, duplicates, and the pool cap.

The resulting preliminary pool \(\mathcal{P}_0\) may still contain many
low-value cliques. LIH therefore performs one warm-up pricing iteration and
ranks non-singleton columns by
\(\bar{c}_p(\lambda^{\mathrm{w}})/|S_p|\), where
\(\lambda^{\mathrm{w}}\) is the warm-up multiplier vector. It then keeps all
singletons and the top \(\eta\) fraction of non-singleton columns. In our
experiments, \(\eta=5\%\) is sufficient to retain nearly all columns used by
optimal partitions on exact-reference instances. The remaining task is to
select disjoint retained columns whose contraction is acyclic.

\begin{algorithm}[!htpb]
\caption{\texttt{RandomGreedyColumns} for candidate-pool generation.}
\label{alg:m7-random-greedy}
\LinesNumbered
\textbf{Input:} SPG $G=(V,E,A)$, vertex weights $w$, initial pool
$\mathcal{P}$, size bound $q$, pool cap $m_{\max}$, rounds
$R_{\mathrm{rg}}$, seed $\xi$, optional deadline $\Delta$\;
\textbf{Output:} Updated candidate pool $\mathcal{P}$\;

initialize the random generator with seed $\xi$\;
$V^{\downarrow}\leftarrow$ vertices sorted by nonincreasing $w_v$\;
\For{$\ell=1,\ldots,R_{\mathrm{rg}}$}
{
    \If{$|\mathcal{P}|\ge m_{\max}$ or deadline $\Delta$ is reached}
    {
        \textbf{break}\;
    }
    $V_{\pi}\leftarrow\texttt{RandomShuffle}(V^{\downarrow})$\;
    $S\leftarrow\emptyset$\;
    \For{each $v\in V_{\pi}$}
    {
        \If{$|S|\ge q$}
        {
            \textbf{break}\;
        }
        \If{$\{u,v\}\in E$, $u\not\leadsto_A v$, and
        $v\not\leadsto_A u$ for every $u\in S$}
        {
            $S\leftarrow S\cup\{v\}$\;
            \If{$S\notin\mathcal{P}$ and $S$ is a feasible clique}
            {
                insert $S$ into $\mathcal{P}$ and record $c(S)$ and
                $a_{vp}$\;
            }
        }
    }
}
\textbf{return} $\mathcal{P}$\;
\end{algorithm}

\inlineheading{Repairing the Relaxed Solution}

The relaxed solution from~\eqref{eq:m7-relaxed-selection} may overlap, miss
vertices, or create a cycle after contraction. LIH therefore uses reduced costs
only as guidance and repairs the selection into a legal partition.

Each candidate column \(p\) is assigned a repair score. Let
\begin{equation*}
    \mathrm{saving}(p)
    =
    \sum_{v\in S_p}w_v-c_p .
\end{equation*}
This value measures the local time saved by executing the vertices of \(S_p\)
together instead of executing them as singleton batches. Let \(\tau(p)\) denote
the insertion order of column \(p\) in the generated pool. The columns are
sorted lexicographically by
\begin{equation*}
    \mathrm{score}(p)
    =
    \left(
    \frac{\bar{c}_p(\lambda)}{|S_p|},
    -\mathrm{saving}(p),
    -|S_p|,
    c_p,
    \tau(p)
    \right).
\end{equation*}
The repair first prefers low reduced cost per vertex, then larger singleton
savings, larger columns, lower execution cost, and earlier generated columns.

The repair scans this ranked list and accepts a column only if two conditions
hold. First, all vertices in the column must still be uncovered. Second, adding
the column to the already accepted columns, while leaving all remaining
uncovered vertices as singleton columns, must keep the contracted directed
graph acyclic. The second condition is a look-ahead feasibility test: it
prevents the repair from accepting a locally attractive clique that would make
the final partition impossible to schedule.

If no multi-vertex column can be accepted, LIH inserts the heaviest remaining
vertex as a singleton, guaranteeing progress and feasibility. After exact
coverage is obtained, \texttt{LocalRefine} tries non-increasing legal vertex
migrations among cliques while preserving clique feasibility and contracted
acyclicity; repeated partition signatures are rejected to prevent cycling.
Algorithm~\ref{alg:m7-refine} summarizes the routine.

\begin{algorithm}[!htpb]
\caption{\texttt{LocalRefine}$(G,w,\mathcal{C})$.}
\label{alg:m7-refine}
\LinesNumbered
\textbf{Input:} SPG $G=(V,E,A)$, vertex weights $w$, current clique
partition $\mathcal{C}$\;
\textbf{Output:} A refined legal clique partition\;

record the current partition signature as visited\;
\While{a legal non-increasing migration is found}
{
    process source cliques in ascending order of $c(C)$\;
    process target cliques in descending order of $c(C)$\;
    process vertices inside each source clique in descending order of $w_v$\;
    \For{each candidate move $x:C_s\rightarrow C_t$}
    {
        $C_s'\leftarrow C_s\setminus\{x\}$,
        $C_t'\leftarrow C_t\cup\{x\}$\;
        let $\mathcal{C}'$ be the partition after replacing $C_s,C_t$
        by $C_s',C_t'$ and removing empty cliques\;
        \If{$C_t'$ is a feasible clique,
        $G/\mathcal{C}'$ is acyclic,
        $c(C_s')+c(C_t')\leq c(C_s)+c(C_t)$, and
        $\mathcal{C}'$ has not been visited}
        {
            accept the move and record the new signature\;
            restart the search from the updated partition\;
        }
    }
}
\textbf{return} $\mathcal{C}$\;
\end{algorithm}

\subsection{Complete Procedure, Complexity, and Limitations}
\label{subsec:m7-complete-algorithm}

Algorithm~\ref{alg:m7-lagrangian} summarizes LIH. The algorithm initializes
the singleton partition, builds the pricing-filtered restricted pool, and then
sets all main-loop multipliers to zero. Each iteration prices columns,
periodically repairs and refines a feasible partition, and updates the
multipliers from the coverage violation.
The repair interval \(h\), step parameters \(\rho,\epsilon\), column-size bound
\(q\), pool-generation parameters \(m_{\max},R_{\mathrm{rg}},\xi,\Delta\), and
screening fraction \(\eta\) control the search.

\begin{algorithm}[!htpb]
\caption{Lagrangian-based Iterative Heuristic (LIH) for SPG clique partitioning.}
\label{alg:m7-lagrangian}
\LinesNumbered
\textbf{Input:} SPG $G=(V,E,A)$, vertex weights $w$, iteration limit $T$,
repair interval $h$, step parameters $\rho,\epsilon$, column-size bound
$q$, pool cap $m_{\max}$, random-greedy rounds $R_{\mathrm{rg}}$, seed
$\xi$, screening fraction $\eta$, optional deadline $\Delta$\;
\textbf{Output:} A legal clique partition $\mathcal{C}^{best}$\;

$\mathcal{C}^{0}\leftarrow\{\{v\}:v\in V\}$\;
$UB\leftarrow\sum_{C\in\mathcal{C}^{0}}c(C)$\;
$\mathcal{C}^{best}\leftarrow\mathcal{C}^{0}$\;
$\mathcal{P}_{\mathrm{sg}}\leftarrow\{\{v\}:v\in V\}$\;
$\mathcal{P}_{0}\leftarrow
\texttt{RandomGreedyColumns}(G,w,\mathcal{P}_{\mathrm{sg}},q,m_{\max},$
$R_{\mathrm{rg}},\xi,\Delta)$\;
run one warm-up pricing iteration on $\mathcal{P}_{0}$ and rank
$\mathcal{P}_{0}\setminus\mathcal{P}_{\mathrm{sg}}$ by
$\bar{c}_p(\lambda^{\mathrm{w}})/|S_p|$\;
$\mathcal{P}\leftarrow
\mathcal{P}_{\mathrm{sg}}\cup
\operatorname{Top}_{\eta}(\mathcal{P}_{0}\setminus\mathcal{P}_{\mathrm{sg}})$\;
$\lambda_v\leftarrow0$, for all $v\in V$\;

\For{$t=1,\ldots,T$}
{
    Compute $\bar{c}_p(\lambda)$ for all $p\in\mathcal{P}$ by Eq.~\eqref{eq:m7-reduced-cost}\;
    $\mathcal{P}^{-}\leftarrow\{p\in\mathcal{P}:\bar{c}_p(\lambda)<0\}$\;
    Compute $\phi(\lambda)$ by Eq.~\eqref{eq:m7-dual-function}\;

    \If{$t=1$, or $t=T$, or $t\bmod h=0$, or $\mathcal{P}^{-}=\emptyset$}
    {
        $\widehat{\mathcal{C}}\leftarrow\texttt{Repair}(\mathcal{P},\bar{c}(\lambda))$\;
        $\widehat{\mathcal{C}}\leftarrow
        \texttt{LocalRefine}(G,w,\widehat{\mathcal{C}})$\;
        \If{$\widehat{\mathcal{C}}$ is legal and
        $\sum_{C\in\widehat{\mathcal{C}}}c(C)<UB$}
        {
            $\mathcal{C}^{best}\leftarrow\widehat{\mathcal{C}}$\;
            $UB\leftarrow\sum_{C\in\widehat{\mathcal{C}}}c(C)$\;
        }
    }

    $h_v\leftarrow\sum_{p\in\mathcal{P}^{-}}a_{vp}$, for all $v\in V$\;
    $g_v\leftarrow h_v-1$, for all $v\in V$\;
    \If{$\|g\|_2^2=0$}
    {
        \textbf{break}\;
    }
    Compute $\alpha_t$ by Eq.~\eqref{eq:m7-polyak-step} or Eq.~\eqref{eq:m7-damped-step}\;
    $\lambda_v\leftarrow\Pi_{[-M,M]}(\lambda_v+\alpha_t g_v)$, for all $v\in V$\;
}
\textbf{return} $\mathcal{C}^{best}$\;
\end{algorithm}

\inlineheading{Complexity and Limitations}

We state the time bound in Lemma~\ref{lem:lih-time}. Let
\(n=|V|\), \(m_0=|\mathcal{P}_0|\) be the preliminary-pool size,
\(m=|\mathcal{P}|\) be the retained working-pool size, \(m_d=|A|\), and let
\(q\) be the maximum candidate-column size. The pricing screen gives
\(m\le n+\lceil\eta(m_0-n)\rceil\). Let \(T\) be the number of Lagrangian
iterations, \(h\) the repair interval,
\(r_{\mathrm{rep}}=O(T/h)\) the number of repair calls, \(b\) the repair scan
width, \(c\le n\) the number of cliques in the current partition, and \(L\)
the number of local-refinement passes after one repair.

\begin{lemma}
\label{lem:lih-time}
Under the above notation, Algorithm~\ref{alg:m7-lagrangian} runs in
\[
    \begin{aligned}
    O\bigl(&R_{\mathrm{rg}}nq
    +m_0(q^2+q+\log m_0)
    +T(mq+n)\\
    &+r_{\mathrm{rep}}\bigl(mq\log m+nb(q+n+m_d)\\
    &\qquad+Lc^2q(q+n+m_d)\bigr)\bigr).
    \end{aligned}
\]
Since \(c\le n\), the refinement term is
\(O(Ln^2q(q+n+m_d))\) in the worst case.
\end{lemma}

\begin{proof}
Candidate generation scans \(R_{\mathrm{rg}}\) randomized vertex orders while
growing cliques of size at most \(q\), and checking, recording, pricing, and
ranking at most \(m_0\) columns gives
\(O(R_{\mathrm{rg}}nq+m_0(q^2+q+\log m_0))\). Each Lagrangian iteration
evaluates reduced costs and relaxed coverage counts over the retained pool,
which costs \(O(mq+n)\). A repair call ranks retained columns, performs at most
\(n\) insertions, and scans up to \(b\) candidates per insertion, giving
\(O(mq\log m+nb(q+n+m_d))\). Local refinement checks source--target clique
pairs and legal vertex moves for \(L\) passes, giving
\(O(Lc^2q(q+n+m_d))\). Multiplying the repair and refinement costs by
\(r_{\mathrm{rep}}\) yields the stated bound.
\end{proof}

The lemma shows the intended design trade-off: pricing over the bounded
retained pool is polynomial and relatively cheap, while repair and refinement
dominate because they repeatedly check contracted-graph feasibility.

The quality of LIH depends on the restricted candidate pool. If a useful clique
is never generated or is removed by screening, the Lagrangian selection step
cannot choose it. This is the central trade-off of the method: restricting
\(\mathcal{P}\) greatly reduces the search complexity, while random greedy
construction and the warm-up pricing screen aim to preserve the columns that
are most valuable for parallel execution. Empirically, retaining all singleton
columns and the top \(5\%\) of non-singleton columns provides a compact pool
that is sufficient for near-optimal solutions in the tested instances.

\section{Evaluation}
\label{sec:exp}

This section evaluates MPE on randomly generated series-parallel task graphs.
We compare solution quality and runtime under varying graph sizes,
compatibility densities, precedence ratios, and task-weight distributions
using three methods.

\begin{itemize}[leftmargin=*]
    \item \textbf{Weighted BB-MGC.} An exact branch-and-bound baseline adapted
    from the mixed graph coloring framework in~\cite{kouider2021bi}. Because
    mixed graph coloring is equivalent to our clique-partitioning formulation
    after complementing the compatibility edge set, this method provides the
    optimal weighted makespan when it finishes within the time limit.
    \item \textbf{Randomized Greedy (RG).} A randomized constructive baseline
    that repeatedly scans a randomized topological order and inserts each
    vertex into one feasible existing batch chosen from a small ranked
    candidate list.
    \item \textbf{LIH.} The proposed Lagrangian-based Iterative Heuristic in
    Section~\ref{sec:method7-lagrangian} and
    Algorithm~\ref{alg:m7-lagrangian}. LIH builds a preliminary candidate
    pool, retains all singleton columns and the top \(5\%\) of non-singleton
    columns after warm-up pricing, and then repairs the priced relaxed
    solution into a legal clique partition.
\end{itemize}

All methods are implemented in C++ and compiled with G++~11.2 using the
\texttt{-O3} option. The experiments are run on a computer equipped with an
AMD Ryzen 7 8745HS processor at 3.80\,GHz and 32\,GB RAM. Each run is
terminated if it does not finish within 1800 seconds.

\subsection{Baselines and Experimental Setup}
\label{subsec:exp-bb-mgc}

The weighted BB-MGC baseline uses the mixed graph coloring view. In MPE,
\(E\) denotes compatibility, whereas in mixed graph coloring undirected edges
denote conflicts. We therefore construct
\[
    G^{c}=(V,\overline{E},A),
\]
where \(\overline{E}=\binom{V}{2}\setminus E\). A color class in \(G^{c}\)
then corresponds to a compatible batch in the original MPE instance.

The branch-and-bound framework in~\cite{kouider2021bi} is unit-time; we keep
its exact coloring search but replace the unit color-class cost by
\(c(C)=\max_{v\in C}w_v\). The feasible coloring space is unchanged, so the
adapted baseline remains exact for the weighted makespan in
\eqref{eq:model-partition-cost}; pruning and incumbent comparisons are also
evaluated with the weighted batch cost.

\inlineheading{Randomized Greedy Baseline}

Algorithm~\ref{alg:exp-rg} gives the randomized greedy baseline. Each round
processes vertices in a randomized topological order and inserts the current
vertex into a feasible existing batch when possible. Candidate batches are
ranked by incremental partition cost, new batch cost, and batch size; RG then
samples uniformly from the best \(K\) candidates. The method remains purely
constructive but can escape some unlucky early insertions.

\begin{algorithm}[!t]
\caption{Randomized Greedy Clique Partitioning (RG).}
\label{alg:exp-rg}
\LinesNumbered
\textbf{Input:} SPG \(G=(V,E,A)\), vertex weights \(w\), round limit \(R\),
candidate width \(K\)\;
\textbf{Output:} A legal clique partition \(\mathcal{C}^{best}\)\;

\(\mathcal{C}^{best}\leftarrow\{\{v\}:v\in V\}\)\;
\For{\(r=1,\ldots,R\)}
{
    draw a randomized topological order \(\pi\) of \((V,A)\)\;
    \(\mathcal{C}\leftarrow\emptyset\)\;
    \For{each \(v\in\pi\)}
    {
        \(\mathcal{F}\leftarrow\) existing batches in \(\mathcal{C}\) that can
        accept \(v\) without violating feasibility or contracted-graph
        acyclicity\;
        \If{\(\mathcal{F}=\emptyset\)}
        {
            append the singleton batch \(\{v\}\) to \(\mathcal{C}\)\;
        }
        \Else
        {
            rank \(C\in\mathcal{F}\) by the increase in \(T(\mathcal{C})\),
            by \(c(C\cup\{v\})\), and then by \(|C|\)\;
            choose one batch uniformly from the first
            \(\min\{K,|\mathcal{F}|\}\) ranked candidates\;
            insert \(v\) into the chosen batch\;
        }
    }
    \If{\(\mathcal{C}\) is legal and
    \(T(\mathcal{C})<T(\mathcal{C}^{best})\)}
    {
        \(\mathcal{C}^{best}\leftarrow\mathcal{C}\)\;
    }
}
\textbf{return} \(\mathcal{C}^{best}\)\;
\end{algorithm}

RG and LIH use the same feasibility checks and objective. RG relies on local
insertions, whereas LIH uses Lagrangian prices to guide a repaired column
selection.

\inlineheading{Instance Generation}

Synthetic SPG instances are generated in the mixed-graph form \(G=(V,E,A)\).

\begin{itemize}[leftmargin=*]
    \item \textbf{Vertices and weights.} Each instance contains \(N_v\)
    vertices. By default, \(N_v=25\). Each vertex weight \(w_v\) is sampled
    from \([5,95]\). We test Uniform, Normal, and Bimodal distributions. The
    default setting is Normal with \(\mu=50\) and \(\sigma=15\). The Bimodal
    setting uses two modes with \(\mu_1=20\), \(\mu_2=80\), and \(\sigma=5\).
    \item \textbf{Compatibility and precedence.} We first sample
    \(N_a=\lfloor D N_v(N_v-1)/2\rfloor\) unordered vertex pairs, where \(D\)
    is the graph density and defaults to \(30\%\). A fraction \(R\) of these
    pairs is converted into directed precedence arcs. The default value is
    \(R=5\%\). Directed arcs are oriented according to a random topological
    order to avoid directed cycles. The remaining unordered pairs are inserted
    into \(E\) as compatibility edges. Instances that violate the required
    acyclicity or SPG structure are regenerated.
\end{itemize}

Table~\ref{tab:exp-defaults} lists the defaults; each experiment group varies
one parameter at a time.

\begin{table}[!t]
\caption{Default parameters for random instance generation.}
\label{tab:exp-defaults}
\centering
\begin{tabular*}{\columnwidth}{@{\extracolsep{\fill}}ll@{}}
\toprule
Parameter & Default value \\
\midrule
Number of vertices \(N_v\) & 25 \\
Graph density \(D\) & \(30\%\) \\
Directed-edge ratio \(R\) & \(5\%\) \\
Weight range & \([5,95]\) \\
Weight distribution & Normal, \(\mu=50,\sigma=15\) \\
Instances per setting & 200 \\
Time limit & 1800 seconds \\
\bottomrule
\end{tabular*}
\end{table}

\begin{table}[!t]
\caption{Experiment groups.}
\label{tab:exp-groups}
\centering
\begin{tabular*}{\columnwidth}{@{\extracolsep{\fill}}lll@{}}
\toprule
Group & Varied parameter & Candidate values \\
\midrule
G1 & Number of vertices \(N_v\) & \(5,10,15,20,25,30,35\) \\
G2 & Graph density \(D\) & \(5\%,10\%,\ldots,50\%\) \\
G3 & Directed-edge ratio \(R\) & \(0\%,5\%,10\%,15\%,20\%\) \\
G4 & Weight distribution & Uniform, Normal, Bimodal \\
\bottomrule
\end{tabular*}
\end{table}

\inlineheading{Evaluation Metrics}

For each instance, we record the makespan \(T(\mathcal{C})\) in
\eqref{eq:model-partition-cost}. When BB-MGC finishes within the time limit,
we report the relative gap of RG and LIH:
\begin{equation*}
    \mathrm{Gap}(M)
    =
    \frac{T_M(G)-T_{\mathrm{BB}}(G)}
         {T_{\mathrm{BB}}(G)} ,
    \quad M\in\{\mathrm{RG},\mathrm{LIH}\},
\end{equation*}
where \(T_{\mathrm{BB}}(G)\) is the optimal weighted BB-MGC makespan. We also
report optimal probability and the compression ratio
\begin{equation*}
    \mathrm{Comp.} =
    \frac{\sum_{v\in V}w_v - T(\mathcal{C})}
         {\sum_{v\in V}w_v}.
\end{equation*}
Runtime is measured as wall-clock time. Instances where BB-MGC times out are
excluded from gap and optimal-probability statistics.

\subsection{Overall Comparison and Node Scaling}
\label{subsec:exp-overall}

\begin{table}[!t]
\caption{Overall comparison across all comparable instances.}
\label{tab:exp-overall}
\centering
\setlength{\tabcolsep}{3pt}
\begin{tabular*}{\columnwidth}{@{\extracolsep{\fill}}lcccc@{}}
\toprule
Method & Opt. prob. & Avg. gap & Median gap & Avg. time \\
\midrule
RG & \(41.16\%\) & \(2.16\%\) & \(0.87\%\) & \(13.75\) ms \\
BB-MGC & \(100.00\%\) & \(0.00\%\) & \(0.00\%\) & \(22.83\) s \\
LIH & \(91.25\%\) & \(0.073\%\) & \(0.00\%\) & \(18.19\) ms \\
\bottomrule
\end{tabular*}
\end{table}

Table~\ref{tab:exp-overall} summarizes all instances with an exact reference.
BB-MGC certifies the optimum but requires \(22.83\) s on average. RG is much
faster at \(13.75\) ms, but matches the optimum in only \(41.16\%\) of the
instances and has an average gap of \(2.16\%\). LIH remains close to RG in
runtime (\(18.19\) ms) while reaching \(91.25\%\) optimal probability and an
average gap of only \(0.073\%\). When LIH is not optimal, its median positive
gap is \(0.56\%\) and its 95th percentile positive gap is \(2.32\%\).

\begin{figure*}[!t]
    \centering
    \includegraphics[width=0.86\textwidth]{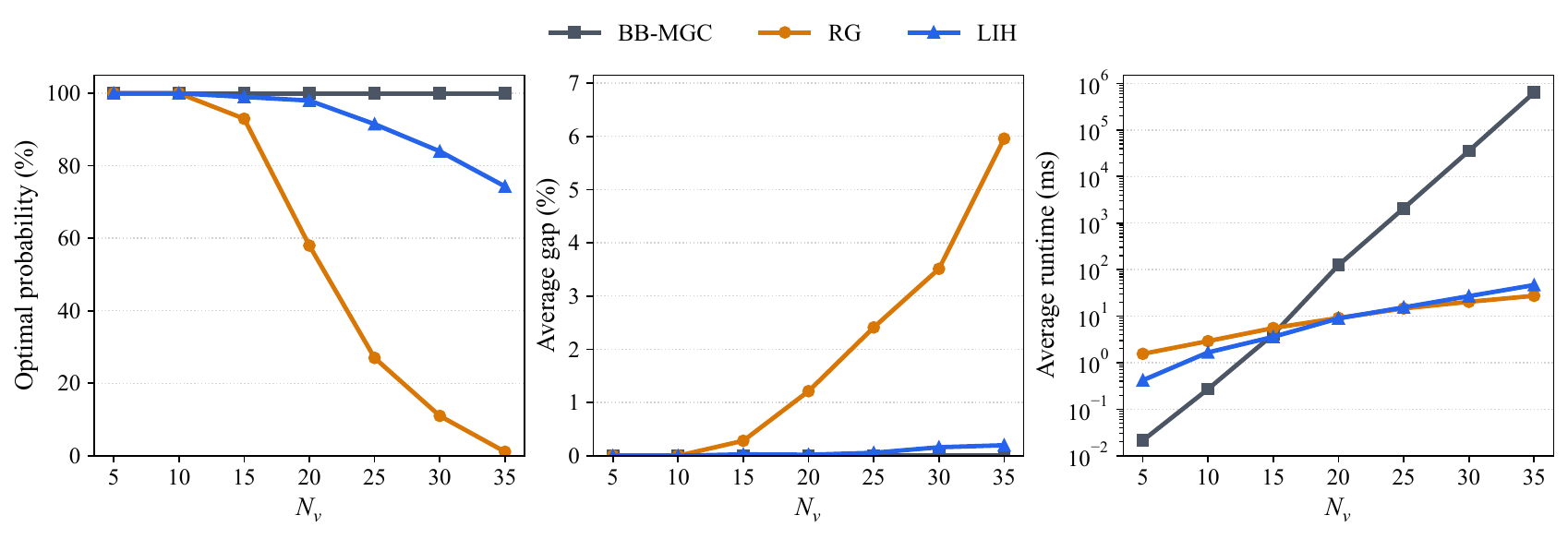}
    \caption{Evaluation results for different numbers of vertices with
configuration: \(D=30\%\), \(R=5\%\), and Normal weights
\((\mu=50,\sigma=15)\).}
    \label{fig:exp-nodes}
\end{figure*}

\inlineheading{Effect of the Number of Vertices}

\begin{table}[!t]
\caption{Solution quality as the number of vertices varies.}
\label{tab:exp-nodes}
\centering
\setlength{\tabcolsep}{2.5pt}
\begin{tabular*}{\columnwidth}{@{\extracolsep{\fill}}rrrrrr@{}}
\toprule
\(N_v\) & Comp. & RG opt. & LIH opt. & RG gap & LIH gap \\
\midrule
5  & 200 & \(100.0\%\) & \(100.0\%\) & \(0.00\%\) & \(0.00\%\) \\
10 & 200 & \(100.0\%\) & \(100.0\%\) & \(0.00\%\) & \(0.00\%\) \\
15 & 200 & \(93.0\%\)  & \(99.0\%\)  & \(0.28\%\) & \(0.026\%\) \\
20 & 200 & \(58.0\%\)  & \(98.0\%\)  & \(1.21\%\) & \(0.019\%\) \\
25 & 200 & \(27.0\%\)  & \(91.5\%\)  & \(2.41\%\) & \(0.056\%\) \\
30 & 200 & \(11.0\%\)  & \(84.0\%\)  & \(3.51\%\) & \(0.160\%\) \\
35 & 183 & \(1.1\%\)   & \(74.3\%\)  & \(5.96\%\) & \(0.198\%\) \\
\bottomrule
\end{tabular*}
\end{table}

Table~\ref{tab:exp-nodes} and Fig.~\ref{fig:exp-nodes} show the node-scaling
trend. The ``Comp.'' column gives the number of instances with a certified
BB-MGC reference; at \(N_v=35\), 17 of 200 BB-MGC runs time out. RG's optimal
probability falls from \(100.0\%\) at \(N_v=10\) to \(1.1\%\) at
\(N_v=35\), while LIH still reaches \(74.3\%\) with an average gap below
\(0.2\%\). BB-MGC runtime grows to \(638.21\) s at \(N_v=35\); RG and LIH
remain in the millisecond range, with averages of \(27.39\) ms and
\(46.65\) ms.

\begin{figure}[!t]
    \centering
    \includegraphics[width=0.86\linewidth]{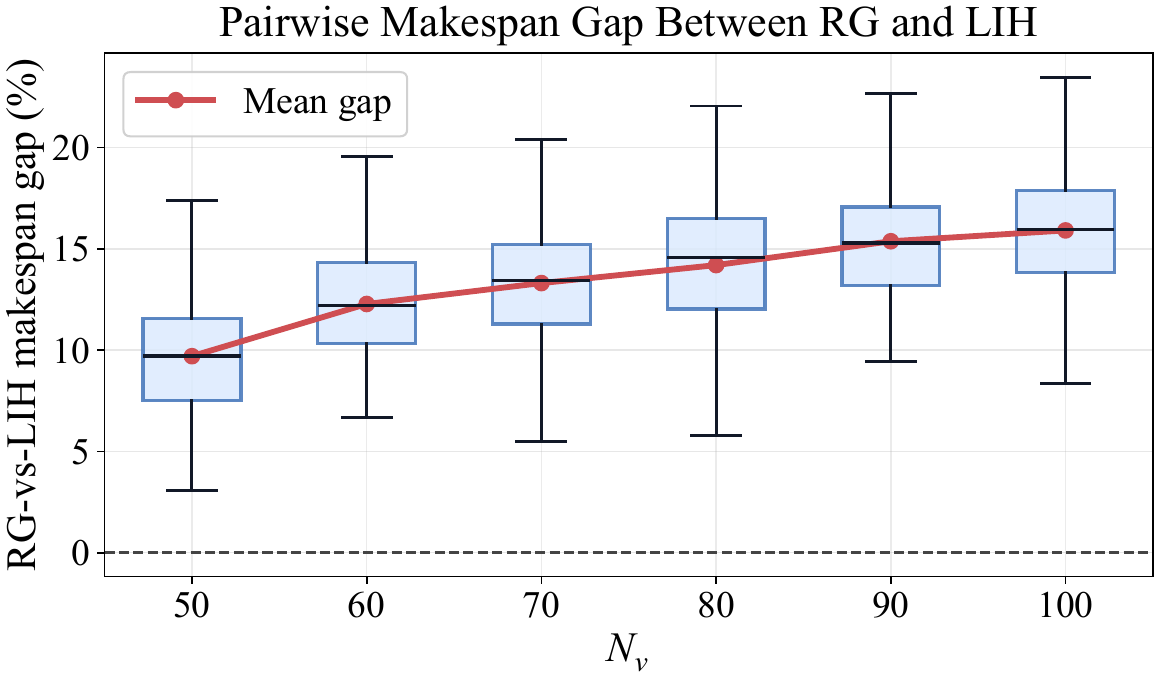}
    \caption{Pairwise makespan gap between RG and LIH on larger random
    instances. The gap is computed as
    \(100(T_{\mathrm{RG}}-T_{\mathrm{LIH}})/T_{\mathrm{LIH}}\), so positive
    values indicate that LIH obtains a smaller makespan.}
    \label{fig:exp-large-nodes-rg-gap}
\end{figure}

\inlineheading{Large-Node Comparison with Randomized Greedy}

For larger instances, where BB-MGC is impractical, we compare LIH directly
with RG for \(N_v=50,60,\ldots,100\), using 200 instances per setting under
the default \(D=30\%\), \(R=5\%\), and Normal-weight configuration.
Figure~\ref{fig:exp-large-nodes-rg-gap} shows that LIH obtains a smaller
makespan in all 1200 pairwise comparisons. The average RG-to-LIH gap increases
from \(9.70\%\) at \(N_v=50\) to \(15.91\%\) at \(N_v=100\), and LIH's
average compression exceeds RG by \(3.39\)--\(4.98\) percentage points. This
supports the exact-reference trend: Lagrangian pricing becomes more valuable
as the clique-combination space grows.

\begin{figure*}[!t]
    \centering
    \includegraphics[width=0.86\textwidth]{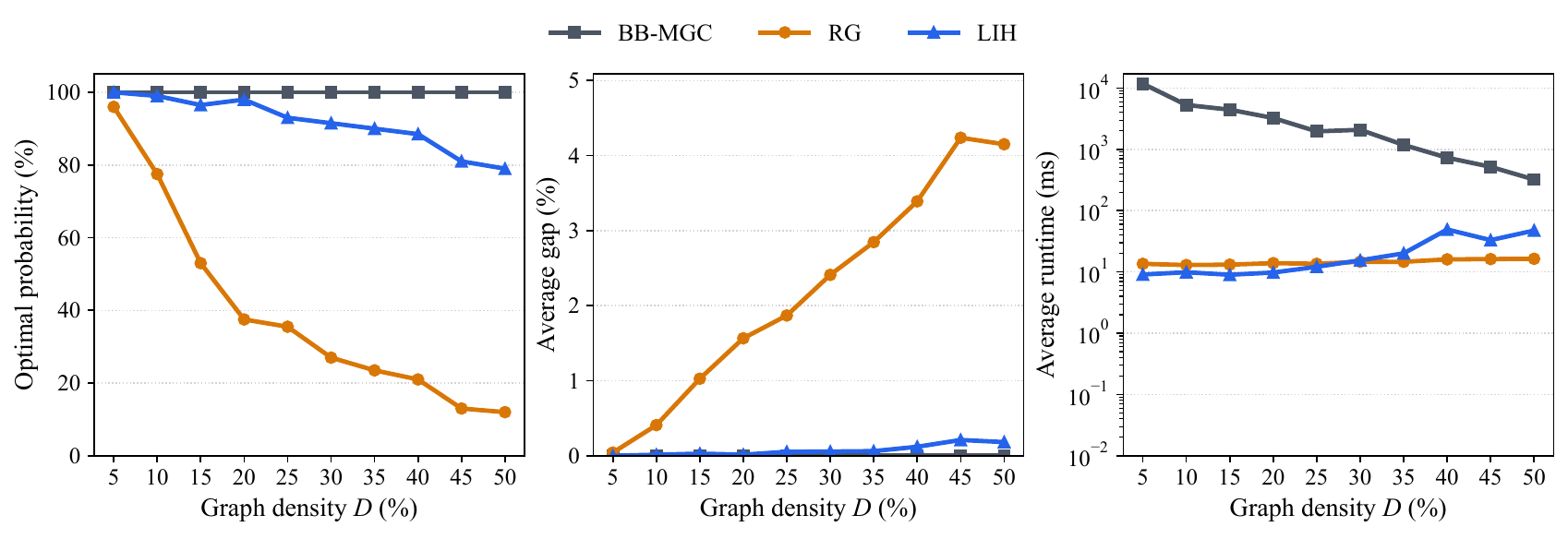}
    \caption{Evaluation results for different graph densities with
configuration: \(N_v=25\), \(R=5\%\), and Normal weights
\((\mu=50,\sigma=15)\).}
    \label{fig:exp-density}
\end{figure*}

\subsection{Sensitivity to Graph and Weight Parameters}
\label{subsec:exp-density}
\inlineheading{Effect of Graph Density}

Figure~\ref{fig:exp-density} reports the effect of compatibility density. As
\(D\) increases from \(5\%\) to \(50\%\), the exact baseline's average
compression rises from \(25.91\%\) to \(66.95\%\), but the grouping decision
also becomes harder. RG's optimal probability drops from \(96.0\%\) to
\(12.0\%\), and its average gap grows from \(0.04\%\) to \(4.15\%\). LIH is
less sensitive: its optimal probability stays between \(79.0\%\) and
\(100.0\%\), and its average gap never exceeds \(0.21\%\).

\begin{figure*}[!t]
    \centering
    \includegraphics[width=0.86\textwidth]{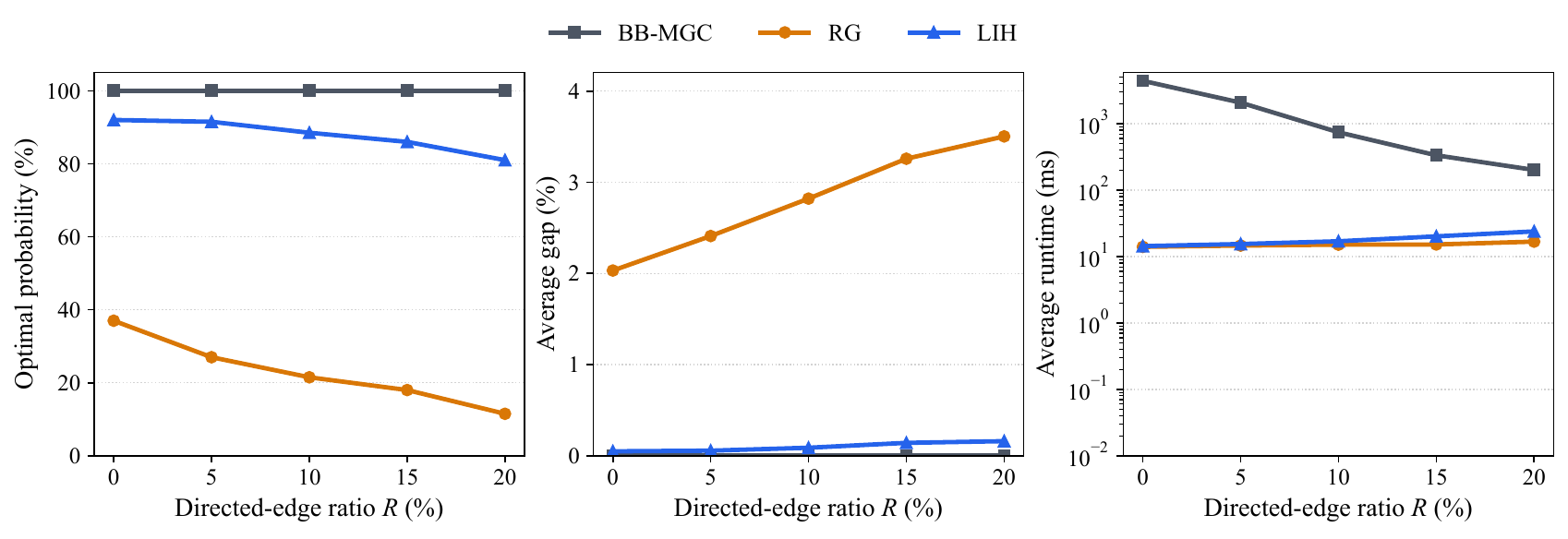}
    \caption{Evaluation results for different directed-edge ratios with
configuration: \(N_v=25\), \(D=30\%\), and Normal weights
\((\mu=50,\sigma=15)\).}
    \label{fig:exp-directed-ratio}
\end{figure*}

\inlineheading{Effect of Directed-Edge Ratio}

Figure~\ref{fig:exp-directed-ratio} isolates precedence constraints. As \(R\)
increases from \(0\%\) to \(20\%\), average exact compression decreases from
\(58.29\%\) to \(52.12\%\), while BB-MGC runtime falls from \(4.44\) s to
\(0.20\) s because ordering constraints reduce the coloring search space. RG's
optimal probability decreases from \(37.0\%\) to \(11.5\%\), with its average
gap rising from \(2.03\%\) to \(3.50\%\). LIH maintains \(81.0\%\)--\(92.0\%\)
optimal probability and keeps the average gap below \(0.17\%\).

\begin{figure*}[!t]
    \centering
    \includegraphics[width=0.86\textwidth]{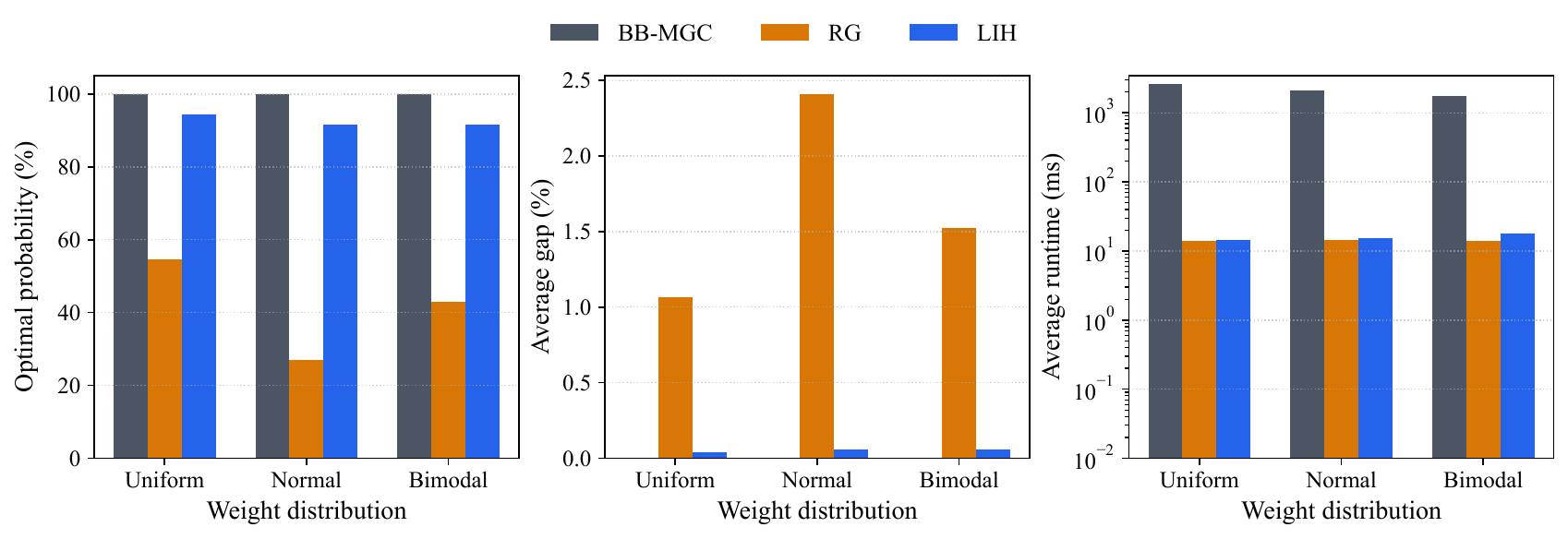}
    \caption{Evaluation results for different weight distributions with
configuration: \(N_v=25\), \(D=30\%\), and \(R=5\%\).}
    \label{fig:exp-weight-distribution}
\end{figure*}

\inlineheading{Effect of Weight Distribution}

Figure~\ref{fig:exp-weight-distribution} compares Uniform, Normal, and Bimodal
weights. RG stays fast but its quality varies: the average gap is
\(1.06\%\), \(2.41\%\), and \(1.53\%\), respectively. LIH is more stable, with
optimal probabilities of \(94.5\%\), \(91.5\%\), and \(91.5\%\), and an
average gap no larger than \(0.056\%\). Reduced-cost pricing therefore captures
weighted savings more reliably than local insertion alone.

\subsection{Discussion}
\label{subsec:exp-discussion}

Overall, BB-MGC is useful for certification but becomes expensive as the
coloring search space expands. RG is lightweight, yet its local insertion
decisions degrade with graph size, density, and precedence complexity. LIH
adds Lagrangian pricing before feasibility repair, giving high optimal
probability, very small gaps, and millisecond-level runtime in the
exact-reference settings. On 50--100-node instances, LIH also consistently
improves over RG in makespan and compression, supporting its use for static
MPE batching when solution quality matters.

\section{Case Study: From PLC Ladder Logic to MPE}
\label{sec:case-study}

This section reports an end-to-end case study that maps PLC ladder programs to
the proposed MPE model, generates staged Verilog, and simulates the resulting
HDL to measure control-cycle latency.

\subsection{PLC Ladder Programs and Dependencies}
\label{subsec:case-plc-semantics}

Programmable Logic Controllers (PLCs) are widely used in industrial automation
and embedded control, where deterministic scan behavior is essential. In each
scan, the controller samples inputs, evaluates the ladder program, and updates
outputs. Ladder Diagram (LD) represents this logic as rungs composed of
contacts and coils: normally open contacts read a signal directly, normally
closed contacts read its complement, and coils update memory variables or
physical outputs.

During rung-by-rung execution, later rungs may observe values produced by
earlier rungs in the same scan. For rung \(R_i\), let
\(\mathrm{Read}(R_i)\) and \(\mathrm{Write}(R_i)\) denote the variables it
reads and writes. If
\(\mathrm{Write}(R_i)\cap \mathrm{Read}(R_j)\neq\emptyset\), then \(R_j\)
depends on \(R_i\). Multiple writes to the same coil also preserve scan order.
Interlocking-signal conflicts, such as complementary NO/NC reads of the same
input or mutually exclusive control actions, are treated as compatibility
constraints even when they do not form direct data dependencies.

Thus, ladder programs contain both local serial chains and independent branches.
MPE exposes the legal within-scan parallelism without changing the observable
PLC scan semantics.

\subsection{SPG Modeling of PLC Ladder Logic}
\label{subsec:case-spg-modeling}

We model a ladder program by constructing the mixed graph \(G=(V,E,A)\) and the
weight vector \(w\) used in Section~\ref{sec:model}. The construction has four
steps.

\begin{enumerate}[leftmargin=*]
    \item \textbf{Task extraction.} Each rung, or each basic logic block inside
    a large rung, becomes a task vertex \(v\in V\). The vertex weight \(w_v\)
    is the estimated execution delay of that rung or logic block.
    \item \textbf{Precedence extraction.} If a value written by \(R_i\) is read
    by \(R_j\), or if two writes must preserve scan order, we add a directed arc
    \((R_i,R_j)\in A\). The reachability relation in \((V,A)\) is later used
    to forbid batching two rungs that are connected by a direct or transitive
    dependence.
    \item \textbf{Compatibility extraction.} Two vertices are connected by an
    undirected edge in \(E\) only when they have no precedence-path conflict,
    no interlocking-signal conflict, and no platform-specific conflict such as
    exclusive I/O access or a shared hardware resource that cannot be
    duplicated.
    \item \textbf{Batch optimization.} The resulting graph is solved as an MPE
    instance. A selected clique represents one legal parallel execution batch,
    and the contracted graph gives the order in which batches are issued.
\end{enumerate}

The implemented front end supports compact ladder descriptions and
PLCOpen-style XML programs~\cite{iacobelli2024detection}. For XML input, each
\texttt{program} POU is converted into one SPG: coils or \texttt{outVariable}
elements become vertices, upstream contacts provide read sets, scan-order
write/read relations become directed arcs, and function blocks become weighted
computation dependencies or HDL module instances.

Table~\ref{tab:case-plc-rungs} gives a representative fragment. \(R_1\) writes
\(C_1\), which is read by \(R_2\) and \(R_3\), so arcs \((R_1,R_2)\) and
\((R_1,R_3)\) are added. \(R_4\) and \(R_5\) read \(I_2\) through
complementary contacts, so their compatibility edge is omitted. The mapping is
shown in Fig.~\ref{fig:case-ladder-spg}.

\begin{table}[!t]
\caption{Example mapping from ladder rungs to MPE constraints.}
\label{tab:case-plc-rungs}
\centering
\setlength{\tabcolsep}{3pt}
\renewcommand{\arraystretch}{1.08}
\begin{tabular*}{\columnwidth}{@{\extracolsep{\fill}}>{\raggedright\arraybackslash}p{0.11\columnwidth}
                >{\raggedright\arraybackslash}p{0.22\columnwidth}
                >{\raggedright\arraybackslash}p{0.16\columnwidth}
                >{\raggedright\arraybackslash}p{0.38\columnwidth}@{}}
\toprule
Rung & Read variables & Write variables & Graph constraint \\
\midrule
\(R_1\) & \(I_1\) (NO) & \(C_1\) & Source rung \\
\(R_2\) & \(C_1,I_4\) (NO) & \(M_1\) & Arc \((R_1,R_2)\) \\
\(R_3\) & \(C_1,I_3\) (NO) & \(M_2\) & Arc \((R_1,R_3)\) \\
\(R_4\) & \(I_2\) (NC) & \(C_2\) & Incompatible with \(R_5\) \\
\(R_5\) & \(I_2\) (NO) & \(M_3\) & Incompatible with \(R_4\) \\
\bottomrule
\end{tabular*}
\end{table}

Assume the rung weights are
\[
    (w_{R_1},w_{R_2},w_{R_3},w_{R_4},w_{R_5})=(4,6,5,3,4).
\]
Serial execution has makespan \(4+6+5+3+4=22\). Applying LIH to this graph
returns the following legal staged batching plan:
\[
    \mathcal{C}=\bigl\{\{R_1,R_4\},\{R_2,R_3,R_5\}\bigr\}.
\]
The two stages preserve the dependencies from \(R_1\) to \(R_2\) and \(R_3\)
and separate the interlocking pair \(R_4,R_5\). Under the staged batching
semantics in Section~\ref{sec:model}, the makespan is
\[
    T(\mathcal{C})
    =
    \max\{4,3\}+\max\{6,5,4\}
    =
    10.
\]
The example illustrates the weighted clique-partitioning objective: each batch
is a feasible clique, and total time sums the slowest task in each batch.

\subsection{FPGA-Oriented Parallel Realization}
\label{subsec:case-fpga}

After LIH produces a legal batching plan, each batch is emitted as one FPGA
execution stage. Contacts become signal reads or negations, LD operations
become combinational logic or function-block instances, and coils are mapped to
intermediate or output registers according to scan semantics.

The generated Verilog module uses a clocked \texttt{start}/\texttt{done}
interface. At the start of a scan, physical inputs are sampled into input-image
registers; a finite-state controller then issues the MPE stages in a
topological order of the contracted graph. All rungs in the current clique are
evaluated in parallel, and their results are registered before the next stage.
This preserves the PLC input/output boundary while extracting parallelism
inside the execution phase.

\subsection{Benchmark and Simulation Results}
\label{subsec:case-benchmark}

The evaluation uses six compact ladder kernels and 30 legitimate PLCOpen XML
programs from the PLC-LD benchmark collection~\cite{iacobelli2024detection}\footnote{\url{https://github.com/UniboSecurityResearch/PLC-LD-dataset}}.
The compact kernels cover representative dependency patterns; the XML programs
test the full parsing, scheduling, HDL generation, and simulation pipeline.

\begin{table*}[!t]
\caption{Per-program comparison of original PLC scan and parallelized HDL.
C-spd. and W-spd. denote cycle speedup and weighted-makespan speedup,
respectively. XML HDL cycles are measured in ModelSim, and compact-kernel HDL
cycles are derived from the same generated stage controller.}
\label{tab:case-all-results}
\centering
\setlength{\tabcolsep}{2pt}
\renewcommand{\arraystretch}{0.90}
\begin{tabular*}{\textwidth}{@{\extracolsep{\fill}}lcccccccccc@{}}
\toprule
Benchmark & \multicolumn{3}{c}{Graph}
& \multicolumn{3}{c}{Cycles}
& \multicolumn{3}{c}{Makespan} \\
\cmidrule(lr){2-4}\cmidrule(lr){5-7}\cmidrule(l){8-10}
& \(|V|\) & \(|A|\) & \(|E|\)
& PLC & HDL & C-spd.
& \(T_s\) & \(T_p\) & W-spd. \\
\midrule
alarm/fault & 10 & 10 & 11 & 10 & 6 & 1.67 & 31 & 18 & 1.72 \\
interlock & 5 & 2 & 7 & 5 & 3 & 1.67 & 22 & 10 & 2.20 \\
motor\_chain & 6 & 4 & 9 & 6 & 4 & 1.50 & 20 & 12 & 1.67 \\
sensor\_vote & 7 & 6 & 4 & 7 & 5 & 1.40 & 24 & 16 & 1.50 \\
pid\_bank & 7 & 6 & 12 & 7 & 4 & 1.75 & 38 & 16 & 2.38 \\
sequence & 8 & 7 & 9 & 8 & 6 & 1.33 & 24 & 18 & 1.33 \\
\midrule
lassignment & 5 & 2 & 8 & 5 & 3 & 1.67 & 46 & 21 & 2.19 \\
lassignment1 & 3 & 0 & 3 & 3 & 2 & 1.50 & 32 & 14 & 2.29 \\
lexit & 7 & 6 & 15 & 7 & 4 & 1.75 & 62 & 29 & 2.14 \\
lstart\_cycle & 5 & 2 & 8 & 5 & 3 & 1.67 & 58 & 27 & 2.15 \\
lstart\_cycle1 & 3 & 0 & 3 & 3 & 2 & 1.50 & 44 & 20 & 2.20 \\
lstart\_eq & 5 & 2 & 8 & 5 & 3 & 1.67 & 49 & 21 & 2.33 \\
lstart\_le & 5 & 2 & 8 & 5 & 3 & 1.67 & 49 & 21 & 2.33 \\
lstart\_le1 & 3 & 0 & 3 & 3 & 2 & 1.50 & 35 & 14 & 2.50 \\
lstart\_lt & 3 & 0 & 3 & 3 & 2 & 1.50 & 35 & 14 & 2.50 \\
lstart\_lt1 & 3 & 0 & 3 & 3 & 2 & 1.50 & 25 & 9 & 2.78 \\
lstop\_eq & 5 & 2 & 8 & 5 & 3 & 1.67 & 49 & 21 & 2.33 \\
lstop\_eq1 & 3 & 0 & 3 & 3 & 2 & 1.50 & 35 & 14 & 2.50 \\
lstop\_ge & 5 & 2 & 8 & 5 & 3 & 1.67 & 49 & 21 & 2.33 \\
lstop\_ge1 & 3 & 0 & 3 & 3 & 2 & 1.50 & 35 & 14 & 2.50 \\
lstop\_gt & 5 & 2 & 8 & 5 & 3 & 1.67 & 49 & 21 & 2.33 \\
lstop\_gt1 & 3 & 0 & 3 & 3 & 2 & 1.50 & 35 & 14 & 2.50 \\
lsub\_function & 6 & 2 & 13 & 6 & 3 & 2.00 & 51 & 21 & 2.43 \\
lsub\_function1 & 6 & 2 & 13 & 6 & 3 & 2.00 & 51 & 21 & 2.43 \\
lsub\_function2 & 6 & 2 & 13 & 6 & 3 & 2.00 & 51 & 21 & 2.43 \\
lsub\_function3 & 4 & 0 & 6 & 4 & 2 & 2.00 & 37 & 14 & 2.64 \\
lsubstitution\_coil & 5 & 2 & 8 & 5 & 3 & 1.67 & 52 & 21 & 2.48 \\
lsubstitution\_coil1 & 3 & 0 & 3 & 3 & 2 & 1.50 & 38 & 14 & 2.71 \\
lsubstitution\_start & 5 & 2 & 8 & 5 & 3 & 1.67 & 51 & 21 & 2.43 \\
lsubstitution\_start1 & 3 & 0 & 3 & 3 & 2 & 1.50 & 37 & 14 & 2.64 \\
lsubstitution\_stop & 5 & 2 & 8 & 5 & 3 & 1.67 & 51 & 21 & 2.43 \\
lsubstitution\_stop1 & 3 & 0 & 3 & 3 & 2 & 1.50 & 37 & 14 & 2.64 \\
lvalue\_filtering & 5 & 2 & 8 & 5 & 3 & 1.67 & 62 & 29 & 2.14 \\
lvalue\_filtering1 & 3 & 0 & 3 & 3 & 2 & 1.50 & 48 & 22 & 2.18 \\
lvalves\_handler & 5 & 2 & 8 & 5 & 3 & 1.67 & 46 & 21 & 2.19 \\
lvalves\_handler1 & 3 & 0 & 3 & 3 & 2 & 1.50 & 32 & 14 & 2.29 \\
\bottomrule
\end{tabular*}
\end{table*}
\begin{figure}[!t]
\centering
\subfloat[Ladder-program fragment.\label{fig:case-ladder}]{%
    \includegraphics[width=0.94\linewidth]{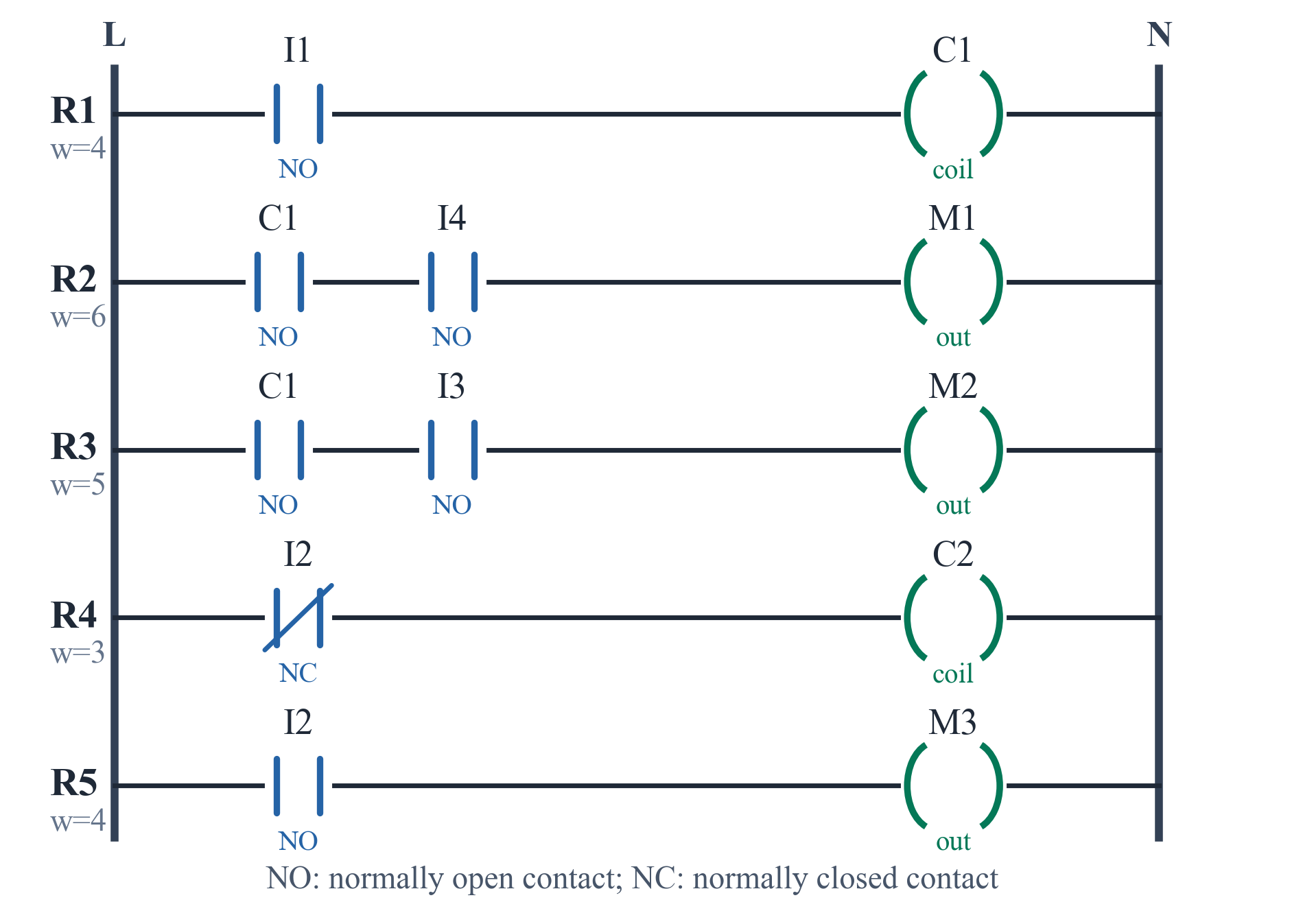}}
\par\medskip
\subfloat[Corresponding SPG.\label{fig:case-spg}]{%
    \includegraphics[width=0.94\linewidth]{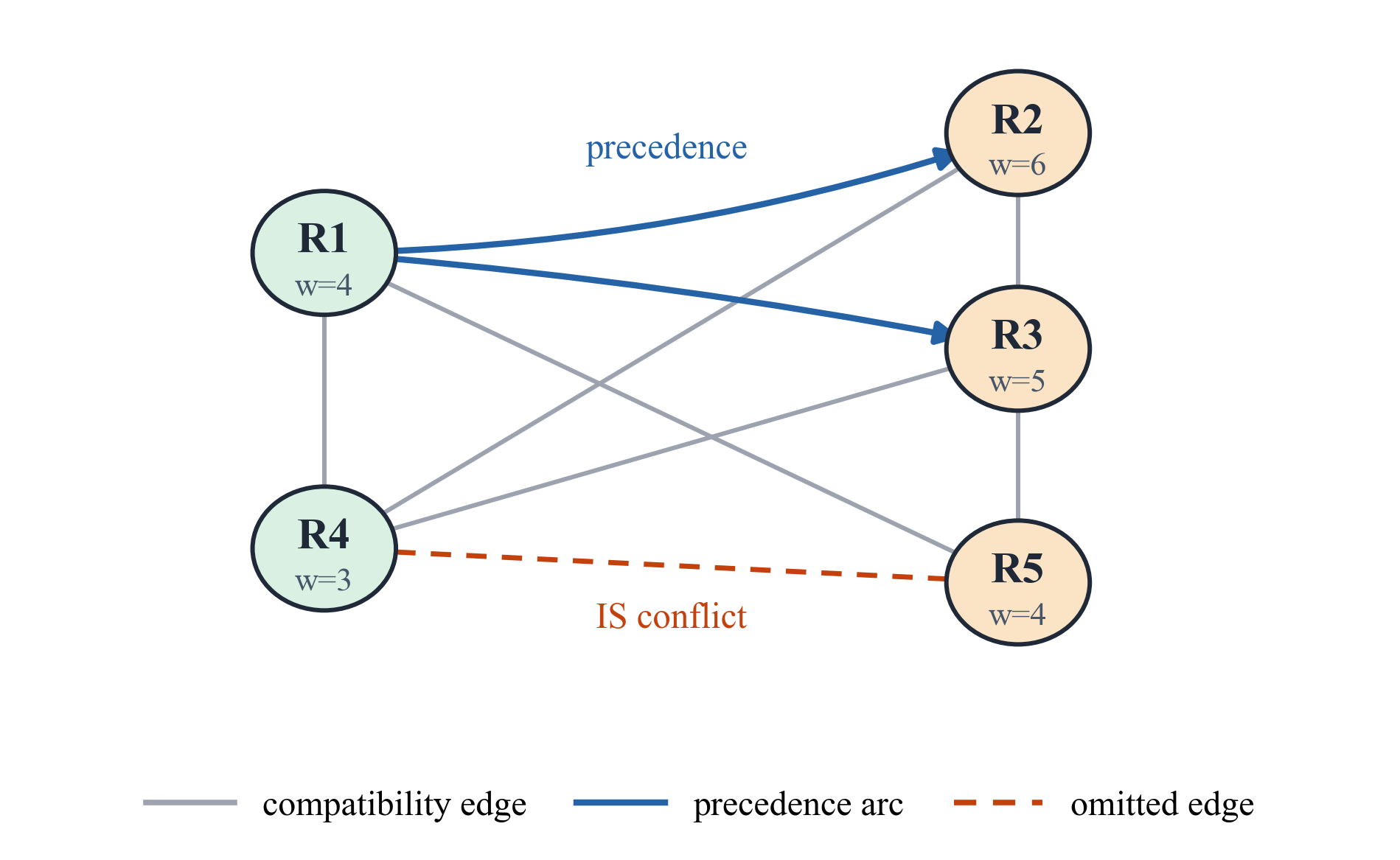}}
\caption{Case-study mapping from a ladder fragment to the corresponding SPG.
In Fig.~\ref{fig:case-spg}, gray edges are compatibility edges, blue arrows
are precedence arcs, and the dashed red segment marks the omitted compatibility
edge caused by an interlocking-signal conflict.}
\label{fig:case-ladder-spg}
\end{figure}

Table~\ref{tab:case-all-results} gives the per-program comparison. Across the
six compact ladder kernels, MPE reduces weighted makespan by \(42.12\%\) on
average; PID-bank and interlock gain most from compatible branches, while
sequence control is limited by longer scan-order chains.

For the PLCOpen XML benchmarks, the pipeline generated one SPG, one LIH
schedule, and one Verilog module per program. All 30 generated modules compiled
with the common HDL library and completed ModelSim simulation. The measured
controller latency falls from \(4.27\) PLC scan cycles to \(2.57\) HDL cycles
on average, giving a \(1.642\times\) cycle speedup and a \(38.54\%\) cycle
reduction. In the weighted MPE metric, the same benchmarks achieve
\(2.399\times\) average speedup, or a \(58.10\%\) staged-makespan reduction.
Thus, the selected cliques are not only scheduling abstractions; they translate
into executable HDL stages that preserve PLC scan semantics.

\section{Related Work}
\label{sec:related}

This section follows the same progression as the paper. We first review
task-graph and series-parallel scheduling models related to the MPE formulation,
then discuss mixed graph coloring and clique partitioning, followed by column
models and Lagrangian relaxation. We finally relate the resulting batching
model to FPGA-based embedded control.

\subsection{Task-Graph and Series-Parallel Scheduling}

Precedence-constrained task-graph scheduling is NP-hard in
general~\cite{ref_dag_nphard}. Classical heterogeneous DAG schedulers such as
HEFT~\cite{topcuoglu2002performance} and PEFT~\cite{arabnejad2013list} map and order
tasks on a fixed processor set while respecting precedence and communication
costs. Recent work has extended this line to FPGA-based heterogeneous real-time
systems~\cite{ahmadi2025messi}, preloaded shared-bus heterogeneous
platforms~\cite{chaudhary2025scheduling}, typed DAG tasks on heterogeneous
multicores~\cite{wu2025partitioned}, and energy-efficient scheduling on
heterogeneous embedded platforms~\cite{hu2023online}. These studies emphasize
resource assignment and schedule construction on specified heterogeneous
resources, whereas MPE abstracts a sufficiently provisioned FPGA stage and asks
which compatible tasks can be batched without violating precedence-induced
acyclicity.

Series-parallel graphs provide additional structure. Their recognition and
decomposition are well studied~\cite{ref_sp_graph1}, and scheduling methods
have exploited this structure for processor-constrained approximation and
static scheduling~\cite{ref_sp_graph2,keller2011peelsched}. In contrast, MPE
forms parallel batches under abundant hardware parallelism, with feasibility
defined by both precedence and compatibility plus acyclicity after contraction.

\subsection{Mixed Graph Coloring and Clique Partitioning}

The MPE formulation is closely related to mixed graph coloring and clique
partitioning. A mixed graph contains both undirected edges and directed arcs,
and mixed graph coloring uses undirected edges to encode conflicts and directed
arcs to encode ordering constraints~\cite{hansen1997mixed}. This viewpoint is
widely used in scheduling models based on disjunctive or mixed graphs. Kouider
and Ait Haddadene~\cite{kouider2021bi} proposed a bi-objective
branch-and-bound algorithm for unit-time job shop scheduling through mixed
graph coloring. We use a weighted adaptation of this exact mixed-coloring
framework as the BB-MGC baseline in the experiments.

Our formulation is related but not identical. In MPE, undirected edges denote
compatibility, so the mixed-coloring instance is obtained by complementing the
edge set. The objective is also weighted: a batch costs the maximum execution
time of its tasks, matching FPGA-style concurrent execution. This gives an
exact BB-MGC baseline for small instances, but its branch-and-bound search
becomes expensive as feasible color classes grow.

\subsection{Column Models and Lagrangian Relaxation}

The exact MPE formulation can be viewed as a set-partitioning model over all
feasible cliques. Such column-based formulations are common in large-scale
combinatorial optimization, but the number of columns can be exponential.
Branch-and-price methods address this difficulty by solving a restricted master
problem and generating useful columns through reduced costs~\cite{barnhart1998branch}.
Lagrangian relaxation provides another classical way to handle hard coupling
constraints by moving them into the objective with
multipliers~\cite{geoffrion1974lagrangean,fisher1981lagrangian}.

LIH follows this philosophy but targets MPE. It constructs a pricing-filtered
pool of SPG-aware cliques, uses Lagrangian multipliers as vertex prices, and
repairs the relaxed selection into an acyclic clique partition. This retains a
global pricing signal without the cost of full branch-and-price.

\subsection{FPGA-Based Embedded Control}

FPGAs are widely used in industrial and safety-critical control because they
provide deterministic timing, fine-grained parallelism, and isolation.
Prior work surveys FPGA design for industrial control~\cite{monmasson2007fpga}
and fault-tolerant SRAM-based FPGA systems~\cite{bernardeschi2015sram}.
Recent studies also accelerate PLC and embedded-control workloads, such as
multi-PID FPGA implementations for reducing PLC scan time~\cite{dhanabalan2022scan},
deterministic parallelization of legacy control code~\cite{hennig2016towards},
and shared-resource effects on real-time microcontrollers~\cite{oliveira2024shared}.

These works motivate hardware acceleration, while our focus is the graph-level
batching problem that decides which compatible tasks can share an FPGA
execution stage. Work on heterogeneous real-time FPGA systems, such as
MESSI~\cite{ahmadi2025messi}, and accelerator-scheduling surveys~\cite{zou2025realtime_survey}
mainly emphasize resource allocation and deadline analysis. We instead study a
sufficiently provisioned setting where dependency and compatibility determine
the available parallelism.

\section{Conclusion}
\label{sec:conclusion}

This paper studied Maximum Parallel Execution (MPE) for series-parallel task
graphs motivated by FPGA acceleration of embedded control programs. We modeled
MPE as a weighted clique-partitioning problem: compatible tasks share a
parallel batch, and selected batches are issued in a topological order of the
contracted precedence graph. The proposed LIH algorithm combines SPG-aware
candidate generation, warm-up pricing screening, Lagrangian pricing,
feasibility repair, and local refinement.

The experiments show that LIH keeps the solution quality close to the exact
BB-MGC reference while remaining in the millisecond range. Across comparable
instances, LIH matches the optimum in \(91.25\%\) of cases with an average gap
of \(0.073\%\), whereas the exact baseline becomes costly as the coloring
search space grows. On 50--100-node instances, LIH consistently improves over
the randomized greedy baseline, indicating that the pricing signal becomes
more useful as the grouping space expands.

The PLC ladder case study further connects the graph model to HDL generation
and cycle-level simulation, showing that selected cliques can be realized as
deterministic FPGA stages. Future work will incorporate explicit hardware
resources, communication overheads, and more detailed FPGA placement or
safety-partitioning requirements.

\bibliographystyle{IEEEtran}
\bibliography{ref}

@inproceedings{valdes1979recognition,
  title={The recognition of series parallel digraphs},
  author={Valdes, Jacobo and Tarjan, Robert E and Lawler, Eugene L},
  booktitle={Proceedings of the eleventh annual ACM symposium on Theory of computing},
  pages={1--12},
  year={1979}
}

@article{nilsson1998real,
  title={Real-time control systems with delays},
  author={Nilsson, Johan},
  journal={PhD Thesis TFRT-1049},
  year={1998},
  publisher={Department of Automatic Control, Lund Institute of Technology (LTH)}
}

@article{mittal2020survey,
  title={A survey of FPGA-based accelerators for convolutional neural networks},
  author={Mittal, Sparsh},
  journal={Neural computing and applications},
  volume={32},
  number={4},
  pages={1109--1139},
  year={2020},
  publisher={Springer}
}

@inproceedings{nurvitadhi2017can,
  title={Can FPGAs beat GPUs in accelerating next-generation deep neural networks?},
  author={Nurvitadhi, Eriko and Venkatesh, Ganesh and Sim, Jaewoong and Marr, Debbie and Huang, Randy and Ong Gee Hock, Jason and Liew, Yeong Tat and Srivatsan, Krishnan and Moss, Duncan and Subhaschandra, Suchit and others},
  booktitle={Proceedings of the 2017 ACM/SIGDA international symposium on field-programmable gate arrays},
  pages={5--14},
  year={2017}
}

@article{monmasson2007fpga,
  title={FPGA design methodology for industrial control systems—A review},
  author={Monmasson, Eric and Cirstea, Marcian N},
  journal={IEEE transactions on industrial electronics},
  volume={54},
  number={4},
  pages={1824--1842},
  year={2007},
  publisher={IEEE}
}

@article{bernardeschi2015sram,
  title={SRAM-based FPGA systems for safety-critical applications: A survey on design standards and proposed methodologies},
  author={Bernardeschi, Cinzia and Cassano, Luca and Domenici, Andrea},
  journal={Journal of Computer Science and Technology},
  volume={30},
  number={2},
  pages={373--390},
  year={2015},
  publisher={Springer}
}

@article{dhanabalan2022scan,
  title={Scan time reduction of plcs by dedicated parallel-execution multiple pid controllers using an fpga},
  author={Dhanabalan, Gnanasekaran and Tamil Selvi, Sankar and Mahdal, Miroslav},
  journal={Sensors},
  volume={22},
  number={12},
  pages={4584},
  year={2022},
  publisher={MDPI}
}

@inproceedings{henriksson2005optimal,
  title={Optimal on-line sampling period assignment for real-time control tasks based on plant state information},
  author={Henriksson, Dan and Cervin, Anton},
  booktitle={Proceedings of the 44th IEEE Conference on Decision and Control},
  pages={4469--4474},
  year={2005},
  organization={IEEE}
}

@inproceedings{hennig2016towards,
  title={Towards parallelizing legacy embedded control software using the LET programming paradigm},
  author={Hennig, Julien and von Hasseln, Hermann and Mohammad, Hassan and Resmerita, Stefan and Lukesch, Stefan and Naderlinger, Andreas},
  booktitle={2016 IEEE Real-Time and Embedded Technology and Applications Symposium (RTAS)},
  pages={1--1},
  year={2016},
  organization={IEEE Computer Soc.}
}

@article{ahmadi2025messi,
  title={MESSI: Task mapping and scheduling strategy for FPGA-based heterogeneous real-time systems},
  author={Ahmadi-Pour, Sallar and Saha, Sangeet and McDonald-Maier, Klaus and Drechsler, Rolf},
  journal={ACM Transactions on Design Automation of Electronic Systems},
  volume={30},
  number={3},
  pages={1--29},
  year={2025},
  publisher={ACM New York, NY}
}

@article{chaudhary2025scheduling,
  title={Scheduling Task Graph Applications on Preloaded Shared-Bus based Heterogeneous Platforms},
  author={Chaudhary, Chhavi and Devaraj, Rajesh and Sarkar, Arnab},
  journal={ACM Transactions on Design Automation of Electronic Systems},
  volume={31},
  number={2},
  pages={1--29},
  year={2025},
  publisher={ACM New York, NY}
}

@article{wu2025partitioned,
  title={Partitioned Scheduling and Analysis for a Typed DAG Task on Heterogeneous Multi-Cores},
  author={Wu, Yulong and Ma, Yehan and Xie, Mingdong and Zhang, Weizhe},
  journal={ACM Transactions on Architecture and Code Optimization},
  volume={22},
  number={3},
  pages={1--24},
  year={2025},
  publisher={ACM New York, NY}
}

@article{hu2023online,
  title={Online energy-efficient scheduling of DAG tasks on heterogeneous embedded platforms},
  author={Hu, Biao and Yang, Xincheng and Zhao, Mingguo},
  journal={Journal of Systems Architecture},
  volume={140},
  pages={102894},
  year={2023},
  publisher={Elsevier}
}

@article{keller2011peelsched,
  title={PEELSCHED: A simple and parallel scheduling algorithm for static taskgraphs},
  author={Keller, J{\"o}rg and Gerhards, Rainer},
  journal={PARS: parallel-algorithmen,-rechnerstrukturen und-systemsoftware},
  volume={28},
  number={1},
  pages={100--109},
  year={2011},
  publisher={Springer}
}

@inproceedings{oliveira2024shared,
  title={Shared resource contention in MCUs: A reality check and the quest for timeliness},
  author={Oliveira, Daniel and Chen, Weifan and Pinto, Sandro and Mancuso, Renato},
  booktitle={36th Euromicro Conference on Real-Time Systems (ECRTS 2024)},
  pages={5--1},
  year={2024},
  organization={Schloss Dagstuhl--Leibniz-Zentrum f{\"u}r Informatik}
}

@article{geoffrion1974lagrangean,
  title={Lagrangean relaxation for integer programming},
  author={Geoffrion, Arthur M},
  journal={Mathematical Programming Study},
  volume={2},
  pages={82--114},
  year={1974}
}

@article{fisher1981lagrangian,
  title={The Lagrangian relaxation method for solving integer programming problems},
  author={Fisher, Marshall L},
  journal={Management science},
  volume={27},
  number={1},
  pages={1--18},
  year={1981},
  publisher={INFORMS}
}

@article{ref_sp_graph1,
  title={The Recognition of Series Parallel Digraphs},
  author={Valdes, Jacobo and Tarjan, Robert E. and Lawler, Eugene L.},
  journal={SIAM Journal on Computing},
  volume={11},
  number={2},
  pages={298--313},
  year={1982},
  publisher={SIAM}
}

@inproceedings{ref_sp_graph2,
  title={Approximation Algorithms for Scheduling on Series-Parallel Graphs},
  author={Elbassioni, Khaled and Regnault, Damien and Srivastav, Anand},
  booktitle={Proceedings of the International Symposium on Algorithms and Computation (ISAAC)},
  pages={97--106},
  year={2009}
}

@book{ref_dag_nphard,
  title={Computers and Intractability: A Guide to the Theory of NP-Completeness},
  author={Garey, Michael R. and Johnson, David S.},
  publisher={W. H. Freeman},
  year={1979}
}

@article{ref_messi,
  title={MESSI: Task Mapping and Scheduling Strategy for FPGA-based Heterogeneous Real-Time Systems},
  author={Fryer, Michael and others},
  journal={IEEE Access},
  volume={11},
  pages={128632--128649},
  year={2023},
  publisher={IEEE}
}

@article{kouider2021bi,
  title={A bi-objective branch-and-bound algorithm for the unit-time job shop scheduling: A mixed graph coloring approach},
  author={Kouider, Ahmed and Ait Haddad{\`e}ne, Hac{\`e}ne},
  journal={Computers \& operations research},
  volume={132},
  pages={105319},
  year={2021},
  publisher={Elsevier}
}

@article{topcuoglu2002performance,
  title={Performance-effective and low-complexity task scheduling for heterogeneous computing},
  author={Topcuoglu, Haluk and Hariri, Salim and Wu, Min-You},
  journal={IEEE transactions on parallel and distributed systems},
  volume={13},
  number={3},
  pages={260--274},
  year={2002},
  publisher={IEEE}
}

@article{arabnejad2013list,
  title={List scheduling algorithm for heterogeneous systems by an optimistic cost table},
  author={Arabnejad, Hamid and Barbosa, Jorge G},
  journal={IEEE transactions on parallel and distributed systems},
  volume={25},
  number={3},
  pages={682--694},
  year={2013},
  publisher={IEEE}
}

@article{hansen1997mixed,
  title={Mixed graph colorings},
  author={Hansen, Pierre and Kuplinsky, Julio and de Werra, Dominique},
  journal={Mathematical Methods of Operations Research},
  volume={45},
  number={1},
  pages={145--160},
  year={1997},
  publisher={Springer}
}

@article{barnhart1998branch,
  title={Branch-and-price: Column generation for solving huge integer programs},
  author={Barnhart, Cynthia and Johnson, Ellis L and Nemhauser, George L and Savelsbergh, Martin WP and Vance, Pamela H},
  journal={Operations research},
  volume={46},
  number={3},
  pages={316--329},
  year={1998},
  publisher={INFORMS}
}

@misc{zou2025realtime_survey,
  title={A Survey of Real-time Scheduling on Accelerator-based Heterogeneous Architecture for Time Critical Applications},
  author={Zou, An and Xu, Yuankai and Ni, Yinchen and Chen, Jintao and Ma, Yehan and Li, Jing and Gill, Christopher and Zhang, Xuan and Jin, Yier},
  year={2025},
  note={arXiv:2505.11970}
}

@article{sankar2021fpga,
  title={FPGA-Based Cost-Effective and Resource Optimized Solution of Predictive Direct Current Control for Power Converters},
  author={Sankar, Deepa and Syamala, Lakshmi and Chembathu Ayyappan, Babu and Kallarackal, Mathew},
  journal={Energies},
  volume={14},
  number={22},
  pages={7669},
  year={2021},
  publisher={MDPI}
}

@article{lucia2017optimized,
  title={Optimized FPGA implementation of model predictive control for embedded systems using high-level synthesis tool},
  author={Lucia, Sergio and Navarro, Denis and Lucia, Oscar and Zometa, Pablo and Findeisen, Rolf},
  journal={IEEE transactions on industrial informatics},
  volume={14},
  number={1},
  pages={137--145},
  year={2017},
  publisher={IEEE}
}

@article{chmiel2023fpga,
  title={FPGA implementation of IEC 61131-3-based hardware-aided timers for programmable logic controllers},
  author={Chmiel, Miroslaw and Czerwinski, Robert and Malcher, Andrzej},
  journal={Electronics},
  volume={12},
  number={20},
  pages={4255},
  year={2023},
  publisher={MDPI}
}

@inproceedings{iacobelli2024detection,
  title={Detection of Ladder Logic Bombs in PLC Control Programs: an Architecture based on Formal Verification},
  author={Iacobelli, Antonio and Rinieri, Lorenzo and Melis, Andrea and Al Sadi, Amir and Prandini, Marco and Callegati, Franco},
  booktitle={2024 IEEE 7th International Conference on Industrial Cyber-Physical Systems (ICPS)},
  pages={1--7},
  year={2024},
  organization={IEEE}
}

\end{document}